\documentclass[aps,prx,twocolumn,showpacs,amsmath,amssymb,superscriptaddress,longbibliography,10pt]{revtex4-2}

\usepackage{amsmath, amsfonts, amsthm, amssymb}
\usepackage{graphicx}
\usepackage{multirow}
\usepackage{bm}
\usepackage{braket}
\usepackage{algorithm}
\usepackage{algpseudocode}
\usepackage{enumerate}
\usepackage{enumitem}

\usepackage[
  colorlinks=true,
  urlcolor=blue,
  linkcolor=blue,
  citecolor=blue
]{hyperref}

\usepackage{thmtools, thm-restate}
\newtheorem{theorem}{Theorem}
\newtheorem{definition}{Definition}
\newtheorem{lemma}{Lemma}

\makeindex

\DeclareMathOperator{\BQP}{\mathsf{BQP}}

\DeclareMathOperator{\Tr}{tr}

\providecommand{\calD}{\ensuremath{\mathcal{D}}}

\providecommand{\calF}{\ensuremath{\mathcal{F}}}

\providecommand{\calL}{\ensuremath{\mathcal{L}}}

\providecommand{\calT}{\ensuremath{\mathcal{T}}}

\providecommand{\calX}{\ensuremath{\mathcal{X}}}

\begin{document}

\title{Exponential quantum advantages in learning quantum observables from classical data}

\author{Riccardo Molteni}
\email{r.molteni@liacs.leidenuniv.nl}
\affiliation{{\small applied Quantum algorithms, Leiden University, 2333 CA Leiden, Netherlands}}
\affiliation{LIACS, Universiteit Leiden, 2333 CA Leiden, Netherlands}

\author{Casper Gyurik}
\email{casper.gyurik@pasqal.com}
\affiliation{{\small applied Quantum algorithms, Leiden University, 2333 CA Leiden, Netherlands}}
\affiliation{LIACS, Universiteit Leiden, 2333 CA Leiden, Netherlands}
\affiliation{Pasqal SaS, 7 rue Léonard de Vinci, 91300 Massy, France}

\author{Vedran Dunjko}
\email{v.dunjko@liacs.leidenuniv.nl}
\affiliation{{\small applied Quantum algorithms, Leiden University, 2333 CA Leiden, Netherlands}}
\affiliation{LIACS, Universiteit Leiden, 2333 CA Leiden, Netherlands}

\begin{abstract}
  Quantum computers are believed to bring computational advantages in simulating quantum many body systems. However, recent works have shown that classical machine learning algorithms are able to predict numerous properties of quantum systems with classical data. Despite various examples of learning tasks with provable quantum advantages being proposed, they all involve cryptographic functions and do not represent any physical scenarios encountered in laboratory settings. In this paper we prove quantum advantages for the physically relevant task of learning quantum observables from classical (measured out) data. We consider two types of observables: first we prove a learning advantage for linear combinations of Pauli strings, then we extend the result for a broader case of unitarily parametrized observables. For each type of observable we delineate the boundaries that separate physically relevant tasks which classical computers can solve using data from quantum measurements, from those where a quantum computer is still necessary for data analysis. Differently from previous works, we base our classical hardness results on the weaker assumption that $\BQP$ hard processes cannot be simulated by polynomial-size classical circuits and provide a non-trivial quantum learning algorithm. Our results shed light on the utility of quantum computers for machine learning problems in the domain of quantum many body physics, thereby suggesting new directions where quantum learning improvements may emerge. 
\end{abstract}

\maketitle

\section{Introduction}

The very first proposed application of quantum computers can be traced back to Feynman's idea of simulating quantum physics on a quantum device. Together with factoring~\cite{shor1999polynomial}, simulation 
of quantum many body systems stands as the clearest example of dramatic advantages of quantum computers~\cite{lloyd1996universal}.  Machine learning is another much newer area where quantum computers are believed to possibly bring advantages in certain learning problems, and in fact there are provable speed ups achieved by a quantum algorithm~\cite{liu2021rigorous,servedio2004equivalences,sweke2021quantum,jerbi2021parametrized,anschuetz2024arbitrary} for specific machine learning problems. Relating back to the original Feynman's idea of simulating quantum physics, machine learning problems in quantum many body physics seem a natural scenario where learning advantages could arise. However, perhaps surprisingly, it was recently shown that access to data exemplifying what the hard-to-compute function does can drastically change the hardness of the computational task questioning the role of quantum computation in a machine learning scenarios~\cite{huang2021power,huang2022provably,lewis2024improved}. If every quantum computation could be replicated classically provided access to data, such results would confine the practical application of quantum computers solely to the data acquisition stage.
This is however not the case, as was shown already in the examples considered in~\cite{servedio2004equivalences,liu2021rigorous}. In such cases the unknown function was cryptographic in nature, and not related to genuine quantum simulation problems. Nonetheless,~\cite{gyurik2023exponential} introduced the first methods for establishing learning separations beyond cryptographic tasks. In particular, they demonstrated that learning problems with provable speed-ups can be constructed from any $\BQP$-complete function, thereby enabling stronger links to physical scenarios. However, the work in~\cite{gyurik2023exponential} had two notable shortcomings. First, the classical non-learnability results were based on relatively strong assumptions about distributional problems, which, while plausible, are less well-understood than their decision problem counterparts.
Second, the study did not introduce any significant quantum learning algorithms for general settings. Instead, quantum learnability was only shown in somewhat artificial scenarios where the concept class was small (polynomially bounded), allowing for the application of straightforward brute-force learning methods.\\
The results presented in this work effectively address both issues, enabling the establishment of clear boundaries between classical and quantum learning algorithms when dealing with data generated by quantum processes. 
 Specifically, we consider learning problems where one is interested in predicting expectation values of an unknown observable from measurements on input quantum states which can either be ground states of local Hamiltonians or time evolved states. As motivation for the learning task, we have in mind experimentally plausible settings where learning an unknown observable may arise, such as in phase classification with an unknown order parameter or when dealing with real devices where the implemented measurement may be influenced by noise or external factors. In this scenarios, our result is a proof of learning advantages for different types of quantum observables. The main contributions of this work are as follows:
 \begin{itemize}
     \item We prove an exponential quantum advantage in learning observables that are formed as a linear combination of local Pauli strings acting on time-evolved or ground states of local Hamiltonians. Importantly, by considering more general learning settings and leveraging results from classical learning theory, we base all of our learning advantage results on the more widely studied and generally accepted assumption that $\BQP$-hard processes cannot be simulated by polynomial-size classical circuit (i.e., $\BQP \not\subseteq \mathsf{P/poly}$), rather than on the less well-understood assumptions related to distributional problems.
     \item Assuming $\BQP \not\subseteq \mathsf{P/poly}$, we show how to construct learning problems that provide a quantum learning advantage in the general case of learning unitarily-parametrized observables, whenever an efficient procedure for learning unknown class of unitaries through query access exists. We then provide a concrete example by connecting to recent results on learning shallow unitaries. 
 \end{itemize}
We also examine the Hamiltonian learning problem where the identification of the target function is demonstrated to be classically easy. We summarize the learning settings and the achieved separations in Table \ref{tab:results}, and provide practical examples of where these learning settings may arise in Section~\ref{sec: conclusion}.

\begin{table*}
   \centering
    \renewcommand{\arraystretch}{1.5} 
    \begin{tabular}{|l|c|c|}
       \hline
        \multicolumn{1}{|c|}{\textbf{Class of observables}} & \multicolumn{1}{c|}{\textbf{Learning Problem}} & \multicolumn{1}{c|}{\textbf{\shortstack{ \\Existence of a \\ Learning Advantage}}}  \\
        \hline
        \multirow{3}{*}{\text{Linear combination of Pauli strings}} & Time-Evolution Problem & Yes  \\
        \cline{2-3}
        & Ground State Problem & Yes  \\
        \cline{2-3}
        & Flipped Concepts & No - Classically easy   \\ 
        \hline
        \multirow{3}{*}{\text{Unitarily-parametrized Observables}} & Learning the observable & Yes   \\
        \cline{2-3}
        & Hamiltonian Learning & No - Classically easy  \\
        \cline{2-3}
        
        & Identifying the concept &  
        
        \shortstack{   \\ Unknown - \\ Classical hardness unknown} \\
        \hline
    \end{tabular}
    \caption{Learning problems investigated in this paper.}
    \label{tab:results}
\end{table*}

\section{Related work}\label{section: related works}
Although we are not the first to prove learning advantages for $\BQP$-complete problems, our results present significant differences compared to previous works.

In~\cite{gyurik2023exponential} the authors first suggested that many learning speed-ups could arise by looking at $\BQP$-complete problems. However, there are at least two crucial differences with our work. Firstly, the classical hardness result in~\cite{gyurik2023exponential} is derived by addressing heuristic complexity classes found in distributional problems. Specifically, it is assumed that there exists an input distribution such that a $\BQP$-complete language $\calL$ cannot be heuristically solved by a classical algorithm, i.e., $(\calL,\calD)\not\subseteq\mathsf{HeurP\slash poly}$. In contrast, in this work we employ a different proof technique and base all of our separation results on the assumption that $\BQP\not\subseteq \mathsf{P\slash poly}$, thus focusing on the more widely investigated complexity classes of decision problems.
Secondly, the learning problems considered in~\cite{gyurik2023exponential} have a concept class of polynomial size, and the quantum algorithm crucially leverages this to find the underlying concept in polynomial time by brute force. In this work we take a step further from their result by considering a physically relevant instantiation of a $\BQP$-complete function. The associated concept classes are continuous and the quantum learning algorithm learns the unknown concept from data through a natural machine learning method.

In recent works by \cite{huang2022provably,lewis2024improved}, significant progresses have been made with classical machine learning algorithms with provable guarantees for learning ground state properties from data, restricting the role of a quantum computer only to the data acquisition phase. In this study, we establish rigorous limitations on the capabilities of classical algorithms when learning quantum observables. A key distinction from their approach is that, in considering ground state problems, we investigate Hamiltonians with a polynomially decaying spectral gap, whereas the algorithms in~\cite{huang2022provably,lewis2024improved} have provable guarantees only for gapped Hamiltonians. The works in~\cite{onorati2023efficient,onorati2023provably} significantly extended the set of classical learnable quantum states to states which lies in the same Lindbladian phase. With respect to their works, our findings establish a definitive boundary on what can be learned through classical methods.

 In \cite{yamasaki2023advantage} the authors proposed a family of supervised learning tasks which present a learning speed-up with an exponentially sized concept class.  Our construction, which was developed independently from~\cite{yamasaki2023advantage}, stems from a physically-motivated context, resulting in the utilization of different techniques for the quantum learning algorithm given the necessity to handle errors in the training samples. Furthermore, their complexity results are, once again, derived by focusing on distributional problems rather than decision problems.
 
 Finally, in~\cite{huang2021information}, the authors provide bounds on the overhead in training samples required by a classical algorithm compared to a quantum algorithm for the task of predicting a fixed observable. In contrast, our results focus primarily on the time efficiency of the learning algorithm. Furthermore, the results in~\cite{huang2021information} assume a fixed target distribution, whereas our setting requires learnability across all input distributions.


\section{Preliminaries}\label{sec:preliminaries}
To make our work more self-contained and accessible, we include the necessary preliminaries from learning theory and complexity theory.
\subsection{PAC learning}
The definition of learnability in this work aligns directly with the widely adopted \textit{probably approximately correct} (PAC) learning framework~\cite{kearns1994introduction,mohri2018foundations}. In the case of supervised learning, a learning problem in the PAC framework is defined by a \textit{concept class} $\mathcal{F}$ which, for each input size $n\in\mathbb{N}$, represent a set functions (or \textit{concepts}) defined from some \textit{input space} $\mathcal{X}_n$ (in this paper we assume $\mathcal{X}_n$ to be $\{0,1\}^n$ ) to some \textit{label set} $\mathcal{Y}_n$ (in this paper we assume $\mathcal{Y}_n$ to be the interval $[-1,1]$). The learning algorithm receives as input information on the unknown target concept $f$  through training samples $T=\{(x_\ell, f(x_\ell))\}_\ell$, where $\bm{x} \in \mathcal{X}_n$ is drawn according to a \textit{target distributions} $\mathcal{D}_n$ over the inputs in $\mathcal{X}_n$.
Finally, the goal of the learning algorithm is to output the description of a function (or \textit{hypothesis}) $h$, which is again a function from $\mathcal{X}_n$ to $\mathcal{Y}_n$,  that is in some sense ``close'' to the target concept $f$ which labels the data in $T$.
\begin{definition}[Efficient PAC learnability]
\label{def:pac}
A concept class $\mathcal{F}$ is \textit{efficiently PAC learnable} if there exists a poly$(1/\epsilon,1/\delta,n)$-time algorithm A such that for all $\epsilon\geq 0 $, all $0\leq\delta\leq 1$, and for any target concept $f$ in $\calF$ and any target distribution $\calD_n$ on $X_n$, if $A$ receives in input a training set $\calT=\{(\bm{x}_\ell,f(\bm{x}_\ell))\}_\ell$ of poly$(1/\epsilon,1/\delta,n)$ size, then with probability at least $1-\delta$ over the random samples in $\calT$ and over the internal randomization of $A$, the learning algorithm $A$ outputs a specification of some hypothesis $h(.)=A(T,\epsilon,\delta,.)$ that satisfies
\begin{equation}\label{eq:real pac}
    \mathsf{Pr}_{\bm{x} \sim \mathcal{D}_n}\big|h(\bm{x})- f(\bm{x})\big| \leq \epsilon.
\end{equation}

\end{definition}
In the above definition, learnability is required for any input distribution. While this might seems a strong requirement, it must be the case in order to ensure the applicability of the theory in real scenarios, where input distributions are often unknown or, in any case, not fixed. 

\subsection{Complexity theory}
Our hardness result for the learning tasks discussed in this work is derived based on the complexity-theoretic assumption that $\BQP\not\subseteq \mathsf{P\slash poly}$. We remark here that, in contrast to previous works~\cite{gyurik2023exponential}, we do not consider the heuristic versions of the classes. Rather, we establish our learning separation by considering their simpler exact decision problem versions. In this section we provide the precise definition of this two complexity classes.

\begin{definition}[\textsf{BQP}]\label{def:BQP}
    A language $\calL$ is in \textsf{BQP} if and only if there exists a polynomial-time uniform family of quantum circuits $\{U_n: n \in \mathbb{N}\}$, such that
\begin{enumerate}
    \item For all $n \in \mathbb{N}$, $U_n$ takes as input an $n$-qubit computational basis state, and outputs one bit obtained by measuring the first qubit in the computational basis.
    \item For all $\bm{x} \in \calL$, the probability that the output of $U_{|x|}$ applied on the input $\bm{x}$ is $1$ is greater or equal to $2/3$.
    \item For all $\bm{x} \notin \calL$, the probability that the output of $U_{|\bm{x}|}$ applied on the input $\bm{x}$ is $0$ is greater or equal to $2/3$.
\end{enumerate}
\end{definition}
The class $\mathsf{P\slash poly}$, instead, captures the class of decision problems solvable by a polynomial time deterministic algorithm equipped with an ``advice'' bitstring. Importantly, the advice depends only on the input size, but it must be otherwise the same for every input $\bm{x}$ of a given size.

\begin{definition}[Polynomial advice~\cite{arora2009computational}]
\label{def:p/poly}
A problem $\calL: \{0,1\}^* \rightarrow \{0,1\}$ is in $\mathsf{P/poly}$ if there exists a polynomial-time classical algorithm $\mathcal{A}$ with the following property: for every $n$ there exists an advice bitstring $\alpha_n \in \{0,1\}^{\mathrm{poly}(n)}$ such that for all $\bm{x} \in \{0,1\}^n$:
\begin{align}
\label{eq:advice}
     \mathcal{A}(\bm{x}, \alpha_n) = \calL(x).
\end{align}

\end{definition}
Equivalently, the class $\mathsf{P\slash poly}$ can be thought as the class of decision problem solvable by a \textit{non-uniform} family of polynomial-size Boolean circuit. That is, for each input size, there exists an efficient circuit that correctly decides the inputs, but the circuits may differ completely for each input size $n$. Finally let us comment on the assumption $\BQP\not\subseteq \mathsf{P\slash poly}$. Although the class $\mathsf{P\slash poly}$ is regarded as very powerful (for instance, it contains undecidable unary languages), it is generally believed that $\BQP$ is not entirely contained within it. For example, if $\BQP\subseteq \mathsf{P\slash poly}$, then factoring and the discrete logarithm problem~\cite{shor1999polynomial} would also be solvable by non-uniform polynomial size classical circuits, which would compromise much of modern cryptography under standard security assumptions~\cite{blum1984generate,gyurik2023exponential}. Another argument against $\BQP\subseteq \mathsf{P\slash poly}$ arises when considering the corresponding sampling problem. If quantum sampling could be done in $\mathsf{SampBPP\slash poly}$, then the polynomial hierarchy would collapse at the fourth level, as discussed in~\cite{aaronson2011computational,gyurik2023exponential,Slack}. Finally, while Adleman's trick allows classical randomness to be simulated with random strings, no analogous method exists for quantum algorithms. In particular, there is no known way to ``extract the quantumness''~\cite{Slack} from a quantum algorithm, leaving the nature of polynomial advice for simulating quantum computers entirely unclear.

\section{The learning problems}\label{section: pauli obs}
We now consider the first class of learning problems addressed in this paper, where the unknown observable is a linear combination of local Pauli strings. Within this framework, we define the following three abstract learning scenarios. In the first case, the input states come from an arbitrary distribution and are time-evolved under a fixed local Hamiltonian before being measured by the (partially) unknown target observable. In the second scenario, the inputs are known local Hamiltonians that time-evolve a fixed initial state before being measured. In the third case, we again consider different Hamiltonians as inputs, but here the learning problem involves learning measurements on their ground states, rather than on time-evolved states. We anticipate that these abstract scenarios can represent a variety of realistic settings where noise or external factors constrain control over the input Hamiltonians and the implemented measurements, we discuss more on this in Section~\ref{sec: conclusion}.

For all the three scenarios, we rigorously model the learning tasks using concept classes within the PAC learning framework introduced in Section~\ref{sec:preliminaries}. We now analyze the first scenario, where the input states are time-evolved under a fixed Hamiltonian for a constant time, however, the input state preparation is not fully controlled. We model this by allowing the input states to be sampled randomly from an unknown distribution of quantum states,  subsequently they undergo a time-evolution before being measured by an unknown observable. After collecting the corresponding measurement outcomes, the goal is to predict the expectation values of the unknown observable on new input states. Let $H$ be a Hamiltonian and fix a constant time $\tau$, we model the \textit{time-evolution learning problem} by the following concept class

\begin{equation}\label{eq:concept class}
   \mathcal{F}_{\mathrm{evolved}}^{H,O}=\{f^{\bm{{\bm{\alpha}}}} (\bm{x})\in \mathbb{R} \;\; | \; \mathbf{{\bm{\alpha}}}\in[-1,1]^m\}
   \end{equation}
 \begin{align*}
     \text{with: } \;\; &f^{\bm{\alpha}}(\bm{x}): \;\;\;\bm{x}\in \{0,1\}^n \rightarrow f^{\bm{\alpha}}(\bm{x})=\text{Tr}[\rho_H(\bm{x}) O(\mathbf{{\bm{\alpha}}})] \\ &O({\bm{\alpha}})=\sum_{i=1}^m {\bm{\alpha}}_i P_i.
\end{align*}
In the above, $\bm{x}$ specifies the initial state $\ket{\bm{x}}$, $\rho_H(\bm{x})$ the constant time evolved state $\rho_H(\bm{x})=U\ket{\bm{x}}\bra{\bm{x}}U^\dagger$ with $U=e^{iH\tau}$ and each $P_i$ is a $k$-local Pauli string where $m$ scales polynomially with $n$. Notice that although we consider binary inputs $\bm{x}\in\{0,1\}^n$, the output of each concept is real-valued. The goal of the ML learning algorithm is to learn a model $h(\bm{x})$ which approximates the unknown concept $f^{\bm{\alpha}}(\bm{x})=\text{Tr}[\rho_H(\bm{x}) O(\mathbf{{\bm{\alpha}}})]$, using as training samples data of the form $\mathcal{T}^\mathbf{{\bm{\alpha}}}_{\epsilon_2} = \{(\bm{x}_\ell, y_\ell)\}_{\ell=1}^N$ where $ y_\ell\approx_{\epsilon_2} f^{\bm{\alpha}}(\bm{x}_\ell)$ is an additive error $\epsilon_2$-approximation of the true expectation value $\text{Tr}[\rho_H(\bm{x}_\ell)O({\bm{\alpha}})]$. Considering datasets with approximated values makes the learning problem closer to real world scenarios. In an idealized setting, the dataset would consist of pairs ${(\bm{x}_\ell,y_\ell)}$, where $y_\ell$ represents the exact expectation value. However, in real experiments, the estimations of these real-valued quantities are obtained by a finite number of state copies, resulting in only approximate estimations of expectation values due to sampling errors.
Formally, we assume a maximum (sampling) error $\epsilon_2$ on the training labels $y_{\ell}$ of our dataset, i.e. $\epsilon_2=\max_\ell|\text{Tr}[\rho_H(\bm{x}_\ell)O({\bm{\alpha}})]-y_\ell|$. It is now possible to formally state our learning condition. Assuming the $\bm{x}_\ell$'s in the training data come from an unknown distribution $\mathcal{D}$, we require that the ML model learns the concept class $\mathcal{F}^{H,O}_{\mathrm{evolved}}$ in the following sense~\footnote{It is important to note that, in contrast to the learning condition in~\cite{gyurik2023exponential}, here we require learnability under any distribution.}.
\begin{definition}[Efficient learning condition]\label{def:learning} A concept class $\mathcal{F}$ is efficiently learnable if there exists a poly$(1/\epsilon,1/\delta,1/\epsilon_2,n)$-time algorithm $A$ such that for all $\epsilon,\epsilon_2>0$ and all $0<\delta<1$, and for any $f^{\bm{\alpha}}$ in $\mathcal{F}$ and any input target distribution $\calD$, if $A$ receives in input a training set $\mathcal{T}^{\bm{\alpha}}_{\epsilon_2}$ of poly$(1/\epsilon,1/\delta,1/\epsilon_2,n)$ size,  $h(.)=A(\mathcal{T}_{\epsilon_2}^{\bm{\alpha}},\epsilon,\delta,.)$ satisfies with probability $1-\delta$:
    \begin{equation}\label{eq:learning}
    \mathbb{E}_{\bm{x}\sim \mathcal{D}}\big[\; |f^{\bm{\alpha}}(\bm{x}) - h(\bm{x})|^2\; \big] \leq \epsilon 
    \end{equation}
\end{definition}
Notice that our learning condition directly follows from the definition of PAC learnability in Def.~\ref{def:pac}, with the only difference being that it is modified to account for errors in the training data, as discussed above.

In our learning task, the learning algorithm only knows the observable through its functional form, mapping $\alpha$ to $O(\alpha)$. Any additional information regarding $O({\bm{\alpha}})$ can exclusively be derived from the training samples within $\mathcal{T}^{\bm{\alpha}}_{\epsilon_2}$.
In particular the vector ${\bm{\alpha}}$, which defines the specific concept in the concept class, is unknown to the learning algorithm. While the classical hardness of the learning problem relies on considering a specific ${\bm{\alpha}}$ for which we prove the evaluation of $f^{\bm{\alpha}}(\bm{x})$ to be hard, the quantum algorithm will use a LASSO regression to infer a parameter $w\in[-1,1]^m$ close to the target ${\bm{\alpha}}$ so that condition (\ref{eq:learning}) is satisfied.

We emphasize here that our learning definition demands that the trained classical model can label new points. This stands in contrast to other settings, e.g. Hamiltonian learning problems, where the task would be identifying the vector ${\bm{\alpha}}$. In the prior case, a learning advantage is more easily established as explained in \cite{gyurik2023exponential} where the hardness of \textit{evaluating} versus \textit{identifying} a concept was discussed. In Section~\ref{sec: relationHL}, we will explore this difference further and present an example of an identification problem closely related to the concept class of Eq.~(\ref{eq:concept class}), which indeed can be solved by a classical algorithm.

We can now state the first result of this paper, namely the existence of a concept class for the time-evolution learning problem learnable by a quantum algorithm but for which no classical algorithm can meet the learning condition of Eq. (\ref{eq:learning}). 

\begin{theorem}[Learning advantage for the time-evolution problem]\label{thm: separation}
 For any $\BQP$-complete language there exists a Hamiltonian $H_{\mathrm{hard}}$ and a set of observables $\{O({\bm{\alpha}})\}_{\bm{\alpha}}$ such that no classical algorithm can efficiently solve the time-evolution learning problem, formalized by the concept class $\mathcal{F}_{\mathrm{evolved}}^{H_{\mathrm{hard}},O}$, in the sense of Def.~\ref{def:learning}, unless $\BQP\subseteq\mathsf{P\slash poly}$. However, there exists a quantum algorithm which learns $\mathcal{F}_{\mathrm{evolved}}^{H_{\mathrm{hard}},O}$ under any input distribution $\mathcal{D}$. 
\end{theorem}

In the remainder of this Section we prove Theorem \ref{thm: separation} by explicitly constructing an example of a provable classically hard time-evolution problem and providing a quantum algorithm with learning guarantees.
We sketch the proof ideas in the following paragraphs of the main text, while the full proofs can be found in the Appendix~\ref{app: classical hardness} and~\ref{app: quantum learnability}.
\\

\subsection{Classical hardness}\label{sec:cl hardness}
Our classical hardness results are based on the assumption $\BQP\not\subseteq\mathsf{P\slash poly}$. As previously stated, this is different from previous works~\cite{gyurik2023exponential}, where classical intractability was established by considering the distributional version of complexity classes. Specifically, the results in~\cite{gyurik2023exponential} explicitly assume the existence of an input distribution $\calD_{\mathrm{hard}}$ for which no classical algorithm can even approximate, in the sense of Eq.~\ref{eq:real pac}, a target function that computes $\BQP$-complete languages. In contrast, this work considers the weaker assumption that $\BQP$-hard languages cannot be correctly decided on every input by polynomial-sized classical circuits. While the assumption considered in~\cite{gyurik2023exponential} implies the assumption considered in this paper, the reverse is not generally true.  
To achieve classical hardness through the more natural and studied assumption $\BQP\not\subseteq\mathsf{P\slash poly}$, we require a stronger but arguably more natural learning condition in Def.~\ref{def:learning} than what present in~\cite{gyurik2023exponential}. Namely, we work in the distribution free PAC framework and require the learning algorithm to succeed for any input distribution. In this way, we can use well-known results from classical learning theory that link the hardness of learning to the hardness of computing a target function. The full proof of classical hardness, as well as quantum learnability, for the learning task can be found in Appendix~\ref{app: classical hardness} and Appendix~\ref{app: quantum learnability} respectively. To prove classical hardness is enough to show that there exists a local Hamiltonian $H_{\mathrm{hard}}$ for which the concepts considered in our learning problem can decide $\BQP$ languages. In fact, by the results in~\cite{schapire1990strength} from classical learning theory, if a target concept is learnable in the sense of Def.~\ref{def:pac}, then there exists a polynomial size classical circuit which efficiently evaluates it on every input point. In particular, this implies that the concept can be computed in $\mathsf{P\slash poly}$.
Therefore, if the concept class $\calF_{\mathrm{evolved}}^{H_{\mathrm{hard}},O}$ is classically learnable, then $\BQP\subseteq\mathsf{P\slash poly}$. As a final note, we observe that although our learning problem is defined over real-valued functions, in our hardness result, the specific Hamiltonian $H_{\mathrm{hard}}$ and the associated observables ensure that the target concept has a binary output. 

\begin{restatable}[Classical hardness of the time-evolution learning problem]{theorem}{hardness} \label{lemma: cl hardness}
For any $\BQP$-complete language, there exists a Hamiltonian $H_{\mathrm{hard}}$ such that no randomized polynomial-time classical algorithm $A_c$ satisfies the learning condition of Def.~\ref{def:learning} for the concept class $\mathcal{F}^{H_{\mathrm{hard}},O}_{\mathrm{evolved}}$, unless $\BQP \subseteq \mathsf{P/poly}$.
\end{restatable}
\begin{proof}[Proof sketch]
The detailed proof of Theorem~\ref{lemma: cl hardness} can be found in the Appendix~\ref{app: classical hardness}, here we give the overall proof strategy.
Let $\mathcal{L}$ be an arbitrary $\BQP$ language. Since $\mathcal{L}\in \BQP$ there exists a quantum circuit $U_{BQP}$ which decides input bitstrings $\bm{x}\in \{0,1\}^n$ correctly on average with respect to $\mathcal{D}$. As shown more rigorously in the Appendix~\ref{app: classical hardness}, measuring the $Z$ operator on the first qubit of the state $U_{BQP}\ket{\bm{x}}$ will output a positive or negative value depending if $\bm{x}\in\mathcal{L}$ or not. Therefore, for the observable $O'=Z\otimes I\otimes...\otimes I$, the quantum model $f^{O'}(\bm{x})=\text{Tr}[O'\rho_{U_{BQP}}(\bm{x})]$, with $\rho_{U_{BQP}}(\bm{x})=U_{BQP}\ket{\bm{x}}\bra{\bm{x}}U^\dagger_{BQP}$, correctly decides every input $\bm{x}$. Finally, using Feynman's idea~\cite{feynman1985quantum,nagaj2010fast} it is possible to construct a local Hamiltonian which time-evolves the initial state $\ket{\bm{x}}$ into $U_{BQP}\ket{\bm{x}}$ in constant time~\cite{childs2004quantum}. Denote this constructed Hamiltonian as $H_{\mathrm{hard}}$. Then the concept $\bm{\alpha'}$, associated to the observable $O({\bm{\alpha}}')=O'$, implements a $\BQP$ computation. The final step is a result in~\cite{schapire1990strength} which guarantees that if a function is learnable under any distribution in the sense of Def.~\ref{def:pac}, then there exists a polynomial size circuit which correctly evaluates it on every input $\bm{x}$. We provide a more detailed explanation of the crucial result contained in~\cite{schapire1990strength} and its consequences in Appendix~\ref{app: classical hardness}. This concludes the proof, as if $\calF^{H_{\mathrm{hard}},O}_{\mathrm{evolved}}$ was learnable, then there would be a polynomial size circuit which evaluates the concept $f^{\alpha'}$. Thus any $\BQP$ language could be decided in $\mathsf{P\slash poly}$.
\end{proof}

\subsection{Quantum learnability}
 To establish the quantum learnability of the classically hard concept class constructed above for the time-evolution problem, we present a quantum algorithm directly. This algorithm accurately predicts any concept in the associated learning problem with high probability when provided with a sufficient number of training data samples. Crucially, we prove that the algorithm meets the learning condition of Def. \ref{def:learning} using only polynomial training samples and running in polynomial time. The central idea involves leveraging the capability of the quantum algorithm to efficiently prepare the time-evolved  states $\rho_{H(\bm{x})}$, for an input local Hamiltonians $H$ and, in particular, for the hard instances of $H_{\mathrm{hard}}$ considered in our hardness result of Lemma \ref{lemma: cl hardness}.

\begin{theorem}[Quantum learnability of the time-evolution learning problem]\label{lemma: q easiness}
There exist an efficient quantum algorithm $A_q$ such that for any concept $f^{\bm{\alpha}}\in \mathcal{F}^{H_{\mathrm{hard}},O}_{\mathrm{evolved}}$ considered in Lemma \ref{lemma: cl hardness}, $A_q$ satisfies the following. 
Given $n,\epsilon,\delta\geq 0$, and any training dataset $T^{{\bm{\alpha}}}_{\epsilon_2}$, with $\epsilon_2\leq\epsilon$, of size 
     \begin{equation}\label{eq:bound}
         N= \mathcal{O}\left(\frac{\log(\mathrm{poly}(n)/\delta)\mathrm{poly}(n)}{\epsilon^2}\right)
     \end{equation}
     with probability $1-\delta$ the quantum algorithm $A_q$ outputs a model $h^*(.)=A(T^{{\bm{\alpha}}}_{\epsilon_2},\epsilon,\delta,.)$ which satisfies the learning condition of Def.~\ref{def:learning}:
     \begin{equation}\mathbb{E}_{\mathbf{x} \sim \mathcal{D}}\big[\; |f^{{\bm{\alpha}}}(\mathbf{x}) - h^*(\bm{x})|^2\; \big] \leq \epsilon 
       \end{equation} 
       where $\mathcal{D}$ is any arbitrary distribution from which the training data is sampled.
\end{theorem}
The idea for the quantum algorithm is the following. For every point $\bm{x}_\ell$ we construct a vector $\phi(\bm{x}_\ell)\in [-1,1]^m$ of expectation values of the single Pauli strings present in $O$ on the time evolved quantum states. The model $h(\bm{x})=w\cdot\phi(\bm{x})$ is then trained on the data samples with a LASSO regression to find a $w^*$ so that the trained model is in agreement with the training samples, i.e. $h^*(\bm{x}_\ell)=w^*\cdot\phi(\bm{x}_\ell)\approx y_\ell$ for any $(\bm{x}_\ell,y_\ell)\sim\mathcal{T}^{\bm{\alpha}}_{\epsilon_2}$. As the $\ell_1$-norm of the optimal $w_{opt}={\bm{\alpha}}$ scales polynomially in $n$, imposing the constraint $||w||_{\ell_1}\leq B$ with $B=\mathcal{O}(\mathrm{poly}(n))$ in the LASSO regression will allow to obtain an error $\epsilon$ in the generalization performance using a training set of at most polynomial size. This is because the generalization error for the LASSO regression is bounded linearly by $B$~\cite{mohri2018foundations}. The description of the quantum algorithm can be found in Algorithm~\ref{algo:quantum}, while we leave the precise analysis of its sample and time complexity in the Appendix~\ref{app: quantum learnability}.

\begin{figure}
\begin{algorithm}[H]
\caption{Quantum Algorithm}\label{algo:quantum}
\begin{algorithmic}
\State \begin{enumerate}
     \item For every training point in $T^\mathbf{{\bm{\alpha}}}_{\epsilon_2} = \{(\bm{x}_\ell,y_\ell)\}_{i=1}^N$ the quantum algorithm prepares poly($n$) copies of the state $\rho_{H}(\bm{x}_\ell)$ and computes the estimates of the expectation values $\langle P_j\rangle_\ell=\text{Tr}[\rho_H(\bm{x}_\ell) P_j] \;\; \forall j=1,...,m$ up to a certain precision $\epsilon_1$. Note that $m$ scales at most polynomially in $n$ as $\{P_j\}_{j=1}^m$ are local observables.
    \item Define the model $h(\bm{x})=w\cdot\phi(\bm{x})$, where $\phi(\bm{x})$ is the vector of the Pauli string expectation values $\phi(\bm{x})= \left[\text{Tr}[\rho_{H}(\bm{x})P_1],...,\text{Tr}[\rho_{H}(\bm{x})P_m]\right]$ computed at Step 1. Then given as hyperparameter a $B\geq0$ the LASSO ML model finds an optimal $w^*$ from the following optimization process:
     \begin{equation}\label{eq: regression}
         \min_{\substack{w\in \mathbb{R}^m\\ ||w||_1\leq B}}\;\; \frac{1}{N}\sum_{l=1}^{N}|w\cdot\phi(\bm{x}_l) - y_l|^2
     \end{equation}
     with $\{(\bm{x}_l,y_l=\text{Tr}[\rho_H(\bm{x}_l)O({\bm{\alpha}})])\}_{l=1}^N$ being the training data. 
     
     Importantly, to meet the learning condition the optimization does not need to be solved exactly, i.e. $w^*={\bm{\alpha}}$. As we will make it clear in the Appendix~\ref{app: quantum learnability}, it is sufficient to obtain a $w^*$ whose training error is $\epsilon_2$ larger than the optimal one.
     \end{enumerate}
     
   \end{algorithmic}  
    
\end{algorithm}
\end{figure}

\subsection{Generalizations: quantum advantages for fixed inputs and for ground state problems}
We showed that the concept class defined in Eq. (\ref{eq:concept class}) leads to a learning advantage. In this case, the physical problem modeled by $\calF_{\mathrm{evolved}}^{H,O}$ assumes the initial input states are drawn from an underlying distribution, while there is precise control over the Hamiltonian governing their evolution. This raises the question of whether learning advantages can be extended to other scenarios, modeled by cases with limited control over the Hamiltonian. We argue that the following scenarios can represent practical examples in Section~\ref{sec: conclusion}.
\vspace{4mm}

\textbf{Fixed inputs, unknown Hamiltonians} A first example that naturally arises from the case discussed above is when a fixed initial state is prepared and time-evolved under a Hamiltonian from a fixed family of Hamiltonians, over which there is no full control. Specifically the Hamiltonians in such a family will be labeled by some input bitstring $\bm{x}$, which for example could parameterize the strength of the coupling interactions. Each initial state will then be evolved by a different Hamiltonian in the family accordingly to the input $\bm{x}$, which comes from an unknown underlying distribution. It is easy to see that the mathematical description of such a learning problem is again defined by the concept class in Eq.(\ref{eq:concept class}). The only change is in the definition of $\rho_{H}(\bm{x})$, now being $\rho_{H}(\bm{x})=U(\bm{x})\ket{0}\bra{0}U(\bm{x})^{\dagger}$ with $H(\bm{x})$ the Hamiltonian associated to the quantum circuit $U(\bm{x})$ by the Feynman construction \cite{feynman1985quantum}. A learning advantage exists for such concept class as well, as the general definition of a language $\mathcal{L}$ in  $\textrm{BQP}$ implies the existence of a family of circuits $\{U(\bm{x})\}_x$ which correctly decides every $\bm{x}\in \mathcal{L}$.  
\vspace{4mm}

\textbf{Ground state problem} As the next example, we consider predicting ground state properties of local Hamiltonians. Here, the states to be measured are the ground states of local Hamiltonians, rather than time-evolved quantum states. Specifically, ground states of input $k$-local Hamiltonians, which belong to a family of Hamiltonians, are prepared (Section~\ref{sec: conclusion} discusses when this is feasible on a quantum computer). However, there is no full control over the coupling parameters, which are instead random values drawn from an underlying distribution. This situation is close, for example, to the case of the random Ising \cite{nattermann1998theory} or random Heisenberg model \cite{oitmaa2001two}.
The mathematical formalization of such a learning problem is the following.
Consider a family of parametrized local Hamiltonians $\mathcal{H}=\{H(\bm{x})\;\; | \;\; \bm{x}\in\{0,1\}^n\}$. Let us define the concept class for the ground state learning problem similarly to the time-evolution case of Eq. (\ref{eq:concept class}), where now the unknown observable $O({\bm{\alpha}})$ is measured on the states $\rho_{H}(\bm{x})$, which correspond to the ground states of the Hamiltonians $H(\bm{x})\in \mathcal{H}$, then we define:
\begin{equation}\label{eq: gs concept class}
    \mathcal{F}^{\mathcal{H},O}_{\mathrm{g. s.}}=\{f^{\bm{\alpha}} (\bm{x})\in\mathbb{R} \;\; | \; \mathbf{{\bm{\alpha}}}\in[-1,1]^m\}
    \end{equation}
\begin{align*}
    \text{with:  } \;\;&f^{\bm{\alpha}}(\bm{x}): \;\;\;\bm{x}\in \{0,1\}^n \rightarrow f^{\bm{\alpha}}(\bm{x})=\text{Tr}[\rho_H(\bm{x})O(\mathbf{{\bm{\alpha}}})]\ \\
    &O({\bm{\alpha}})=\sum_{i=1}^m {\bm{\alpha}}_i P_i.
\end{align*}
Considering the training data $\mathcal{T}^\mathbf{{\bm{\alpha}}}_{\epsilon_2} = \{(\bm{x}_\ell, y_\ell\approx_{\epsilon_2} f^{\bm{\alpha}}(\bm{x}_\ell))\}_{\ell=1}^N$, the learning condition remains the same as in Definition \ref{def:learning}.
From the hardness result of Lemma \ref{lemma: cl hardness}, we obtain the following Theorem
\begin{theorem}[Learning advantage for the ground state learning problem]\label{thm: separation gs}
For any $\BQP$-complete language there exists a family of Hamiltonians $\mathcal{H}_{\mathrm{hard}}$ such that no classical algorithm can learn the ground state problem, formalized by learning the concept class $\mathcal{F}^{\mathcal{H}_{\mathrm{hard}},O}_{\mathrm{g. s.}}$ in the sense of Def.\ref{def:learning}, unless $\BQP\subseteq\mathsf{P\slash poly}$. However, there exists a quantum algorithm which learns $\mathcal{F}^{\mathcal{H}_{\mathrm{hard}},O}_{\mathrm{g. s.}}$ under any input distribution $\mathcal{D}$. 
\end{theorem}

\begin{proof}
The existence of a class of Hamiltonian $\mathcal{H}_{\mathrm{hard}}$ for which the ground state problem is classically hard to learn is guaranteed by the argument above regarding the hardness of time-evolution case.   
The structure of the proof is exactly the same of the proof for Theorem \ref{thm: separation} rigorously written in the Appendix~\ref{app: classical hardness}.
The only missing step for the ground state version is that now the states $\rho_{H}(\bm{x})$ we are considering in Eq. (\ref{eq: gs concept class}) are ground state of Hamiltonians $H(\bm{x})$ and not time evolved states. However, using the Kitaev's construction \cite{kitaev2002classical,kempe20033}, it is possible to create for any $\bm{x}\in\{0,1\}^n$ and $U$ a local Hamiltonian $H(\bm{x})$ such that its ground state will have a large overlap with $U\ket{\bm{x}}$, with $U$ an arbitrary quantum circuit with a polynomial depth. This completes our proof of classical hardness as we consider the family of Hamiltonian $\mathcal{H}_{\mathrm{hard}}$ to exactly be the set of such $\{H(\bm{x})\}_x$ with $U$ implementing a $\BQP$-complete computation~\footnote{More rigorously it means that for every $n$, we consider $U=U^n_{BQP}$ from the family $\{U^n_{BQP}\}_n$ of quantum circuits which correctly decide a $\BQP$-complete language $\mathcal{L}$}. Same as before, it is still the case that there is at least one concept which can not be evaluated by polynomially sized classical circuits. Finally, also the quantum algorithm with learning guarantees for the concept class $\mathcal{F}^{\mathcal{H}_{\mathrm{hard}},O}_{\mathrm{g. s.}}$ closely follows Algorithm \ref{algo:quantum}. The only missing point to prove here is that the ground states $\rho_{\mathcal{H}_{\mathrm{hard}}}(\bm{x})$ are easily preparable on a quantum computer. Recall that the class of hard Hamiltonian $\mathcal{H}_{\mathrm{hard}}$ considered in the hardness result of the ground state problem are the one derived from the Kitaev's circuit-to-Hamiltonian construction from a $\BQP$-complete circuit. It is well known that those Hamiltonians present a $\Omega(1/\text{poly}(n))$ gap (in contrast to the classically learnable Hamiltonians in \cite{huang2022provably}) and it is possible to construct the ground state $\rho_{\mathcal{H}_{\mathrm{hard}}}$ from the description of the corresponding Kitaev's Hamiltonian, known to the learner through the description of the concept class and the input $\bm{x}$. This then concludes our proof for the quantum learnability of $\mathcal{F}^{\mathcal{H}_{\mathrm{hard}},O}_{\mathrm{g. s.}}$, thus completing the entire proof of Theorem \ref{thm: separation gs} 
\end{proof}

It is important to note that although our separation results, both for the time-evolution and the ground state versions of the problem, are proven using Kitaev's Hamiltonians, analogous results hold for a broader and more physical class of Hamiltonians. Specifically, since we rely on the assumption that $\BQP \not\subseteq \mathsf{P/poly}$, any quantum process that cannot be simulated by polynomial-sized circuits gives rise to a learning problem of the kind we introduced, with a provable quantum learning advantage. In Section~\ref{sec: conclusion} we list examples of physical Hamiltonians to which our results apply, covering both the time-evolution and ground state versions of the problem.

We note that the ground state learning problem defined by the class in Eq. (\ref{eq: gs concept class}) closely resembles the machine learning problem studied in \cite{huang2022provably} and \cite{lewis2024improved}, where the authors demonstrated that a classical algorithm could solve the task, provided the Hamiltonian has a constant gap. Since the Hamiltonians considered in our results have a polynomially decaying gap, this imposes a constraint on the family of local Hamiltonians $\mathcal{H}$ that one would need to consider in order to prove a quantum advantage in learning. 
\vspace{4mm}

\textbf{The ``flipped case''} As a final remark, it is interesting to observe that if we consider the case of fixed input state, which could be a fixed initial state $\rho_0(\bm{x})=\ket{\bm{x}}\bra{\bm{x}}$ in the time-evolution scenario modeled by the concept class in Eq. (\ref{eq:concept class}) or a single ground state $\rho_{H}(\bm{x})$ of a fixed Hamiltonian $H(\bm{x})$ in the ground state problem, then the learning problem becomes trivially classically easy. Such learning scenario is formally equivalent to considering a ``flipped concept'' of the ones defined above. Let us for example consider the task of learning an unknown quantum process from many measurements, in this case the learning problem is still modeled by the concept class $\mathcal{F}^{H,O}_{\mathrm{evolved}}$ of Eq. (\ref{eq:concept class}) with the difference that now the role of $\bm{x}$ and ${\bm{\alpha}}$ are switched.
Namely, the concepts are defined as $f^{\bm{x}}({\bm{\alpha}})$ so that the their expressions remain the same of $f^{\bm{\alpha}}(\bm{x})$. The difference lies in the labeling, where $\bm{x}$ now denotes the concept while ${\bm{\alpha}}$ represents the input vectors. As $\bm{x}$ is constant for an instance of the learning problem, specified by the concept that generates the data,
the training samples are measurements of the same quantum state $\rho_{H(\bm{x})}$ with different observables corresponding to different ${\bm{\alpha}}$.
Since $O({\bm{\alpha}})=\sum_{j}{\bm{\alpha}}_j P_j$, we can make use of data to solve the linear system and obtain the expectation values of each local Pauli string $\text{Tr}[\rho_{H_{\mathrm{hard}}}(\bm{x})P_i]$. It then becomes easy to extrapolate the value of $\text{Tr}[\rho_{H_{\mathrm{hard}}}(\bm{x})O({\bm{\alpha}})]$ for every new ${\bm{\alpha}}$.

\subsection{Generalization of the quantum learning mechanism to quantum kernels}

In \cite{liu2021rigorous}, considerable effort was invested to construct a task with provable learning speed-up where the quantum learning model is somewhat generic (related to natural QML models studied in literature), while still being capable of learning a classically unlearnable task.
In our construction, it is possible to view the entire learning process as a quantum kernel setting, similarly to the approach in \cite{yamasaki2023advantage}. 
However, in standard kernel approaches, especially those stemming from support vector machines, the optimization process is solved in the dual formulation where the hypothesis function is expressed as a linear combination of kernel functions. In contrast, the LASSO optimization employed here solves the optimization in the primal form with a constrain on the $\ell_1$ norm. This is an issue making our learner not \textit{technically} a quantum kernel. While the LASSO formulation does not straightforwardly convert into a kernel method, one could attempt to address the regression problem outlined in Equation (\ref{eq: regression}) using an alternative optimization approach that supports a kernel solution, such as kernel ridge regression. In this case the optimization is done for vectors with bounded $\ell_2$ norms, nevertheless there exist bounds on the generalization performance of such procedures as well \cite{mohri2018foundations}.

\section{Generalization to observables parametrized by unitaries}\label{section: unitaries}

We have demonstrated the existence of a quantum advantage for the learning problem of predicting $k$-local observables from time evolved states and from ground states.
Critically, we considered observables of the type $O({\bm{\alpha}})=\sum_i {\bm{\alpha}}_i P_i$ and we exploited the linear structure of $O({\bm{\alpha}})$ to ensure quantum learnability through LASSO regression. It is however natural to ask if quantum learnability can be achieved for other types of observables, while maintaining the classical hardness. In this section we consider the far more general case where the unknown observable is parameterized through a unitary matrix, i.e. $O({\bm{\alpha}})=W({\bm{\alpha}}) O W^\dagger({\bm{\alpha}})$ where $O$ is an hermitian matrix.  
Our findings in this scenario will be of two kinds. First we show as a general result that \textit{for every method which learns a unitary $W$ given query access on a known distribution of input quantum states there exists a learning problem which exhibits a classical-quantum advantage}. Then we concretize the general result by presenting a constructed example of a learning problem defined by a class of unitary-parameterized observables, showcasing a provable speed-up. Before stating our findings, let us properly introduce the learning problem under consideration. Imagine a scenario where observables are measured on evolved quantum states, but there is no control over the entire evolution, with a portion of it remaining unknown. Specifically, measurements are taken on input-dependent states $\ket{\psi(\bm{x})} = W({\bm{\alpha}}) U(\bm{x}) \ket{0}$ for an unknown fixed ${\bm{\alpha}}$. The goal is to predict expectation values on states $\ket{\psi(\bm{x}')}$ for new inputs $\bm{x}'$.
Concretely such scenario corresponds to the following concept class:
\begin{equation}\label{eq: concept unitary}
\mathcal{M}_{U,W,O}=\{f^{\bm{\alpha}}(\bm{x})\in\mathbb{R}\;\;|\;\; {\bm{\alpha}}\in [-1,1]^m\}\end{equation}
\begin{align*}
\text{with: } &\;\; f^{\bm{\alpha}}(\bm{x}): \;x\in\mathcal{X}\subseteq\{0,1\}^n\rightarrow \text{Tr}[\rho_{U}(\bm{x})O({\bm{\alpha}})]\in \mathbb{R} \\
&O({\bm{\alpha}})=W({\bm{\alpha}}) O W^\dagger ({\bm{\alpha}}). 
\end{align*}
where $\rho_U(\bm{x})=U(\bm{x})\ket{0}\bra{0}U^\dagger(\bm{x})$ and $O$ is a hermitian matrix.
Given as training data $\mathcal{T}^{\bm{\alpha}}=\{(\bm{x}_\ell,y_\ell)\}_{\ell=1}^N$ with $\mathbb{E}[y_\ell]=\text{Tr}[\rho_{U}(\bm{x}_\ell)O({\bm{\alpha}})]$, the goal of the learning algorithm is again to satisfy the learning condition of Def.~\ref{def:learning}. We are now ready to state our main general result of this Section:
\begin{theorem}[Informal version]\label{thm: separation W}
Every (non-adaptive)~\footnote{Non-adaptive means that the algorithm probes the unitary with some input states and measures
some observables which are chosen independently, meaning they do not have any choice which depends on
previous outcomes.}  learning algorithm $\mathcal{A}_W$ for learning a unitary $W(\bm{\alpha})\in\{W(\bm{\alpha})\}_{\bm{\alpha}}$, where the probe states $\{\ket{\psi_l}\}_l$ and observables $\{Q_m\}_m$ come from discrete sets $S=\{\ket{\psi_{\ell}}\}_{\ell}$ and $Q=\{Q_m\}_m$ (or they can be discretized with controllable error), induces a classical-input-classical-output learning problem with a quantum-classical learning advantage.
 
\end{theorem}

In other words, Theorem \ref{thm: separation W} guarantees that whenever there is an efficient method to learn a unitary from query access on a set of arbitrary input states and observables it is also possible to construct a learning problem described by the concept class in Eq.~(\ref{eq: concept unitary}) which exhibits a learning speed-up. We will more mathematically formalize the statement of Theorem~\ref{thm: separation W} and provide a rigorous proof in the Appendix~\ref{app:unitaries}. The nontrivial part of our result lies in the fact that in the majority of the works present in the literature, which provide efficient algorithms to learn a unitary, the set of required probe states $\{\ket{\psi_\ell}\}_\ell$ is restricted to specific quantum states, often classically describable (e.g. stabilizer states). On such a set of states, our hardness results from the previous section can not directly be applied as a classical learning algorithm could simply prepare those states as the quantum algorithm would do. 

The detailed proof of Theorem \ref{thm: separation W} can be found in the Appendix~\ref{app:unitaries}, here we outline the main ideas.

\begin{proof} [\textit{Proof sketch}]

    The main idea of the proof is to construct a concept class of the kind of Eq.(\ref{eq: concept unitary}) which can be learned by a quantum algorithm using the learning algorithm $\mathcal{A}_W$, while maintaining classical harndess. Consider the circuit in Figure \ref{fig:unitary separation}. Let $U(\bm{x})$ be a unitary that, depending on the first bit $x_1$ of the input $\bm{x}\in\{0,1\}^n$, prepares the $\ket{\psi^{n_S}_U(\bm{x})}$ on the $n_S$ qubit register in two different ways. If $x_1=0$, then $\ket{\psi^{n_S}_U(\bm{x})}$ is exactly one of the states $\ket{\psi_\ell}\in S$, labeled by the final $n_S$ bits $\bm{x_S}\in\{0,1\}^{n_S}$ of the input bitstring $\bm{x}$, i.e. $\bm{x_S}=x_{n-n_S}...x_n$. If $x_1=1$, then $U(\bm{x})$ prepares the state $\ket{\psi^{n_S}_U(\bm{x})}$ as the result of a $\BQP$-complete computation. Regarding the observable, we define a controlled unitary $V_A$, controlled by the first $1+n_Q$ qubits register, such that when $x_1=0$ $V_A$ acts on the target $n_S$ qubit register by rotating the $n_S$ qubit measurement operator $O$ into one of the $Q_m$ in the set $Q$ required by the learning algorithm. The description of which $Q_m$ the unitary $V_A$ implements is contained in the $n_Q$ bitstring $\bm{x_Q}\in\{0,1\}^{n_Q}$ with $\bm{x_Q}=x_2x_3...x_{n_Q+1}$. When $x_1=1$, $V_A$ acts as the identity matrix on the $n_S$ qubit register.
    Suppose now the input bits are sampled from an arbitrary distribution $\calD_i\in\{\calD_i\}_i$ of the following kind:
    \begin{itemize}
    \item The first bit $x_1$ of $\bm{x}$ is randomly selected with equal probability between 0 and 1.
     \item If $x_1=0$ then the other $n-1$ bit $x_2x_3...x_n$ are sampled from the required distribution~\footnote{Note that the distribution of probes states and measurements required by $\mathcal{A}_W$ to learn the unitary $W(\bm{\alpha})$ could be very complicated involving perhaps a joint distribution on probe states and measurements. Nevertheless every target distribution can be obtained by post-processing of samples from the uniform distribution, which can be done coherently by taking uniformly random input bitstrings.} by $\mathcal{A}_W$ to learn the unitary $W(\bm{\alpha})$.
    \item If $x_1=1$ then the $n_S$ bits in $\bm{x_S}$ follow any possible distribution over $\bm{x_S}\in\{0,1\}^{n_S}$. For each $i$, the distribution over $\bm{x_S}\in\{0,1\}^{n_S}$ defines the specific overall distribution $\calD_i$.
    \end{itemize}

Taking $W(\bm{0})=I^{\otimes n_S}$ for $\bm{\alpha}=\bm{0}$ and $O=I\otimes I^{\otimes n_Q}\otimes Z \otimes I^{n_S -1}$, it is clear that the concept $f^0=\text{Tr}[\rho_U(\mathbf{x})V_AOV^\dagger_A]$ cannot be learned by any classical algorithm.
In particular, by the result of our Theorem~\ref{lemma: cl hardness} any classical algorithm can not meet the learning condition of Def.~\ref{def:learning} on more than half of the inputs $\bm{x}$ coming from any distribution $\mathcal{D}_i$, namely it is guaranteed to fail for the inputs for which $x_1=1$. 
To prove quantum learnability we notice that for every $\bm{\alpha}$ the dataset $\mathcal{T}^{{\bm{\alpha}}}$ associated to the concept class $\mathcal{M}_{U,W,O}$ contains exactly the pairs of state and measurement outcomes needed by the learning algorithm $\mathcal{A}_W$ in order to learn the unitary $W(\bm{\alpha})$. Namely, half of the training samples, characterized by $x_1=0$ in their input $\bm{x}$, allows the algorithm to recover the unknown $W(\bm{\alpha})$. Therefore the quantum algorithm is able to learn the unknown observable $O({\bm{\alpha}})$ and evaluate it on each quantum state associated to every input $\bm{x}\in\{0,1\}^n$. 
\end{proof}

\subsection{Learning advantages for shallow unitaries}
Theorem~\ref{thm: separation W} presents a method for constructing learning problems that demonstrate a provable quantum advantage in learning unitarily parametrized observables of the form $O({\bm{\alpha}})=W({\bm{\alpha}}) O W^\dagger({\bm{\alpha}})$. This result crucially relies on the existence of a learning scheme capable of recovering the unitary $W({\bm{\alpha}})$ based on output measurements performed on an arbitrary set of input states. As a concrete example where such a learning scheme is known to exist, and thus our result applies, we consider the case where the family of unitaries  $\{W(\bm{\alpha})\}_{\bm{\alpha}}$ consists of shallow circuits. In this case, recent results in the literature~\cite{huang2024learning} provide a procedure to learn few bodies observables from measurement of them on the single qubit stabilizer states. As shallow circuits preserve the locality of the observable measured, we can directly apply the result in~\cite{huang2024learning} to construct a quantum learning algorithm for the concept class in Eq.~(\ref{eq: concept unitary}) where the concepts are related to observables $O({\bm{\alpha}})=W({\bm{\alpha}}) O W^\dagger({\bm{\alpha}})$, with $W({\bm{\alpha}}) $ a shallow unitary. In particular, in the Appendix~\ref{app:unitaries} we prove the following corollary of Theorem~\ref{thm: separation W}.

\begin{restatable}[Learning advantage for shallow unitaries]{corollary}{shallowadv}\label{lemma: learning unitary}
  There exists a family of parametrized  unitaries $\{U(\bm{x})\}_x$ and parametrized shallow circuits $\{W({\bm{\alpha}})\}_{\bm{\alpha}}$, a measurement $O$ and a set of distributions $\{\mathcal{D}_i\}_i$ over $\bm{x}\in\{0,1\}^n$ such that the concept class $\mathcal{M}_{U,W,O}$ is not classically learnable with respect to the set of input distributions $\{\mathcal{D}_i\}_i$, unless $\BQP\subseteq\mathsf{P\slash poly}$. However, there exists a quantum algorithm which learns $\mathcal{M}_{U,W,O}$ on the input distributions $\{\mathcal{D}_i\}_i$. 
\end{restatable}
\begin{proof}[Proof sketch]
 The complete proof of Corollary~\ref{lemma: learning unitary} is provided in Appendix~\ref{app:unitaries}. The proof follows the same approach as in Theorem~\ref{thm: separation W}, namely constructing a family of input distributions $\{\calD_i\}_i$ that allows a quantum algorithm to recover the target unitary $W(\bm{\alpha})$, while ensuring that the learning problem remains intractable for any classical algorithm. Specifically, this is achieved by designing a model that prepares stabilizer states for half of the inputs and performs a $\BQP$-hard computation for the other half. The quantum learning algorithm can then utilize the training data containing stabilizer states to learn the unitary $W(\bm{\alpha})$ by applying the procedure in~\cite{huang2024learning}. Importantly, it will also be able to evaluate the learned $O(\bm{\alpha})$ on the other half of input data.
\end{proof}

\begin{figure*}
    \centering
    \includegraphics[width=0.8\linewidth]{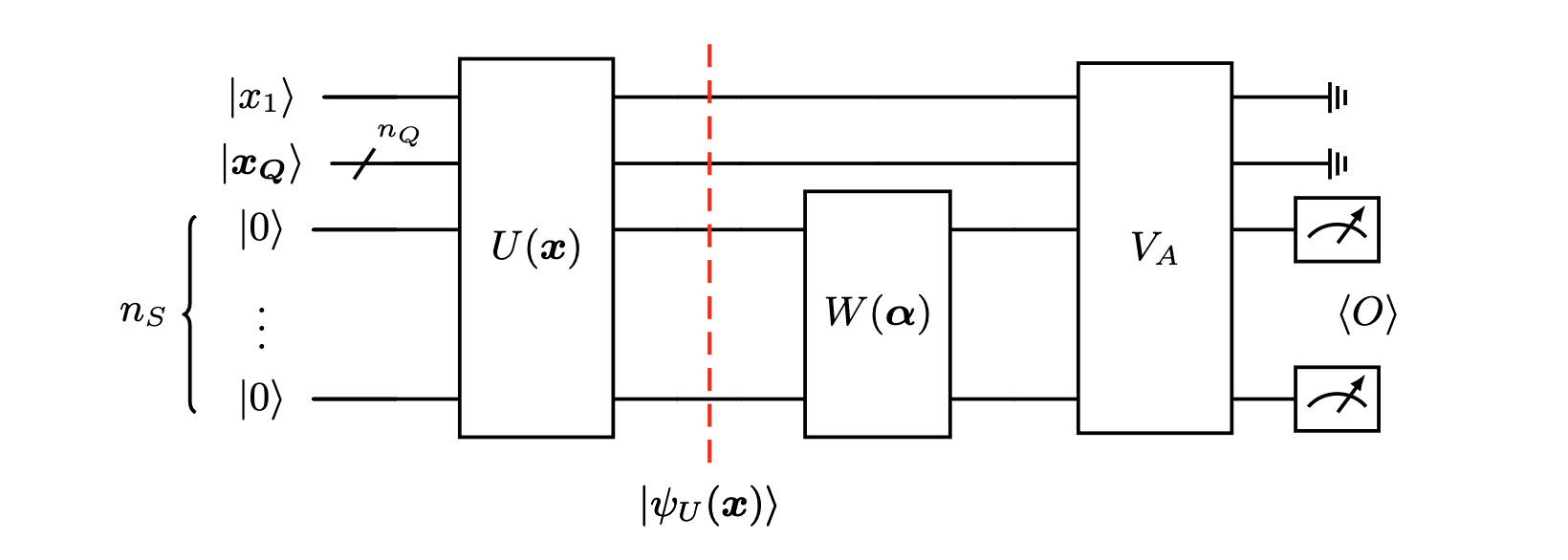}
     \caption{Quantum model which exhibits a learning speed-up for the concept class $\mathcal{M}_{U,W,O}$. The unitary $U(\bm{x})$ prepares the state $\ket{\psi_U(\bm{x})}=\ket{x_1}\otimes\ket{\bm{x_Q}}\otimes\ket{\psi^{n_S}_U(\bm{x_S})}$ where the form of $\ket{\psi^{n_S}_U(\bm{x_S})}$ depends on $x_1$, the first bit of each input $\bm{x}\in\{0,1\}^n$. If $x_1=0$, then $\ket{\psi^{n_S}_U(\bm{x_S})}=\ket{\psi_{\bm{x_S}}}\in S$ is the $n_S$ qubit state described by the bitstring  $\bm{x_S}=x_2x_3...x_{n_S}$ of the set $S$ of polynomially describable quantum states needed to learn $W(\bm{\alpha})$. If $x_1=1$, then $\ket{\psi^{n_S}_U(\bm{x_S})}$ is the quantum state which encodes the $\BQP$-complete computation which decides the $n_S$ input bits $\bm{x_S}$, considered as input of a $\BQP$-complete language $\mathcal{L}$ over $\bm{x_S}\in\{0,1\}^{n_S}$. $W({\bm{\alpha}})$ is an unknown parametrized unitary, in order to prove classical hardness it is sufficient to consider $W({\bm{\alpha}}=0)=I^{\otimes n_S}$ and the measurement operator to be $O=Z\otimes I \otimes ... \otimes I$ on the $n_S$ qubits register. The unitary $V_A$ rotates the measurement operator $O$ so that, when $x_1=0$, the final measurement is an operator $Q_{\bm{x_Q}}$ described by the bitstring $\bm{x_Q}$. This provides the right training samples to learn $W(\bm{\alpha})$, as explained in the proof of Theorem~\ref{thm: separation W}.
    }
  \label{fig:unitary separation}
\end{figure*}

\subsection{Relationship to Hamiltonian learning}\label{sec: relationHL} 

In order to further elucidate the lines between settings with and without classical/quantum learning advantages, we can establish a connection between the task of learning Hamiltonians from time-evolved quantum states and learning unitarily-parametrized observables. The setting of the Hamiltonian learning problem is the following. Given an unknown local Hamiltonian in the form of $H(\bm{\lambda})=\sum_i \lambda_i P_i$, the objective of the Hamiltonian learning procedure is to recover $\bm{\lambda}$. In the version of the problem we consider, we are given a black box which implements the time-evolution under the unknown Hamiltonian $H(\bm{\lambda})$ on any arbitrary quantum state. The black box action on arbitrary inputs $\rho$ is given by
\begin{equation}
    \rho\rightarrow U\rho U^\dagger
\end{equation}
with $U=e^{i tH(\bm{\lambda})}$, where the evolution time $t$ is known to us. In \cite{haah2022optimal} the authors provide a classical algorithm which learns the unknown $\bm{\lambda}$ from the expectation values of the Pauli string $P_i$ on polynomially many copies of the evolved state $U\rho U^\dagger$, for particular choices of initial states $\rho$'s. We can rephrase the Hamiltonian learning task in terms of a learning problem with many similarities to the ones we considered before. Concretely we define the corresponding concept class
\begin{equation}\label{eq: concept HL}
\mathcal{F}_{\bm{\lambda}}=\{f^{\bm{\lambda}}(\bm{x})\in\mathbb{R} \;\;|\;\; \bm{\lambda}\in [-1,1]^m\}
\end{equation}
\begin{align*}
\text{with:  }& \;\; f^{\bm{\lambda}}(\bm{x}): \;\bm{x}\in\mathcal{X}\subseteq\{0,1\}^n\rightarrow \text{Tr}[\rho(\bm{x})O(\bm{\lambda})] \\
&O(\bm{\lambda})=\sum_{i=1}^m U(\bm{\lambda})^\dagger P_i U(\bm{\lambda}). 
\end{align*}
where $\bm{x}$ describes the input state $\rho(\bm{x})$, and $U=e^{itH(\bm{\lambda})}$.
We note that the specific choice of the observable $O(\bm{\lambda}) = \sum_{i=1}^m U(\bm{\lambda})^\dagger P_i U(\bm{\lambda})$ is made to align with the work in~\cite{haah2022optimal}, where the learning algorithm solves the task by directly measuring the Pauli operators $\{P_i\}_i$ after evolving the input states with $U = e^{i t H(\bm{\lambda})}$. In this sense, by considering $\mathcal{T}^{\bm{\lambda}} = {(\bm{x}_\ell, \text{Tr}[\rho(\bm{x}_\ell)O(\bm{\lambda})])}_\ell$ as training data, the connection to the Hamiltonian learning task of \cite{haah2022optimal} becomes evident if we demand that the learning algorithm needs only to \textit{identify}, rather than \textit{evaluate}, the correct concept that generated the data. It is worth noting that the concept class $\mathcal{F}_{\bm{\lambda}}$ closely resembles the concept class for the time-evolution problem described earlier, for which a learning advantage was demonstrated. The main difference is that in $\mathcal{F}_{\bm{\lambda}}$ the unknown observable is no longer a linear combination of Pauli strings, instead each Pauli string in $O(\bm{\lambda})$ is parameterized by a unitary. As we are requiring to only identify the correct concept and considering that the authors in \cite{haah2022optimal} provide an algorithm capable of solving the Hamiltonian Learning problem in polynomial time (for sufficiently small $t$), one might suppose that the classical difficulty associated with the learning advantage presented in this work arises solely from the challenge of evaluating the correct concept rather than from the task of identifying it. This would answer another question left open in \cite{gyurik2023exponential}, at least in the case of concepts labeled by observables with an unknown unitary parameterization. However it must be noted that for the Hamiltonian learning algorithm in \cite{haah2022optimal} to work, the input states $\rho(\bm{x})$ in Eq.(\ref{eq: concept HL}) are of a very particular and simple form and do not come from a distribution $\mathcal{D}$ of $\BQP$-complete quantum states $\rho_{H_{\mathrm{hard}}}(\bm{x})$ as we assumed in our classical hardness result.
For general input distributions, the hardness of identification thus remains elusive. In this regard it is important to note that the task of identification can take various forms. For instance, here we are focusing on identifying the concept within the concept class that generated the data, and in this case, it is generally unclear when this setting allows a classical-efficient solution.
In a more general scenario one could ask if the learning algorithm is able to identify a model from a hypothesis class that differs from the concept class but that it is still capable of achieving the learning condition. Colloquially here the question becomes: can traditional machine learning determine which quantum circuit could accurately label the data, even though classical computers lack the capability to evaluate such quantum circuits. 
In \cite{gyurik2023exponential} it was proven that in this broader context provable speed-ups in identification are no longer possible. This is because the classical algorithm can successfully identify a sophisticated quantum circuit that carries out the entire quantum machine learning process as its subroutine, thereby delegating the learning aspect to the labeling function itself, prior to evaluating a data point. 
Thus the task of identification is indeed only interesting when the hypothesis class is somehow restricted, which is often the case in practically relevant learning scenarios.

\section{Discussion}\label{sec: conclusion}
In this paper, we propose a machine learning task that demonstrates a quantum learning advantage under the assumption that quantum processes cannot be simulated by polynomial-size classical circuits (i.e., $\BQP\not\subseteq\mathsf{P/poly}$). The learning problem involves reproducing a partially unknown observable from classical, measured-out data. While classical learning theory guarantees the computational hardness of the task for input states originating from certain quantum processes, we have developed a quantum learning algorithm capable of solving it for any input distribution. In this discussion, we consider whether this abstract learning problem can reasonably model physically relevant scenarios. We identify three key components for our framework: \textit{\romannumeral 1}) the preparation of a ``hard'' quantum state from a large family, indeed ensuring that $U$ specifies a classically-hard but quantum realizable operation~\footnote{For our purposes, it is sufficient (though not necessary) to have a parameterized heralded state $\rho(x)$ such that measuring an observable $O$ on it produces the function $x\mapsto\mathrm{Tr}[\rho(x)O]$, which cannot be implemented using polynomial-sized circuits.} (see below for details); \textit{\romannumeral 2}) the heralded yet uncontrolled nature of the selection of $U(x)$, reflecting the fact that data is generated under an unknown distribution over $x$ in the PAC learning framework; and \textit{\romannumeral 3}) the task of learning unknown observables itself.

We next discuss the extent to which these three criteria correspond to natural learning tasks. 
\begin{enumerate}[ wide=0pt,label=\textit{\roman*)},labelindent=3mm]
    \item \textbf{Hard state preparation.} 
    While in our proof we deal with the contrived Kitaev Hamiltonian, this is only for technical simplicity. In general, all that is required is a scenario where a family of parametrized ``hard'' states ($\ket{\psi(x)} = U(x)\ket{0}$) is measured, and where there is strong evidence to suggest that computing the map $x\mapsto\bra{\psi(x)}O\ket{\psi(x)}$ (where $x$ is the description of the state in the family) is beyond the capabilities of classical polynomial-size circuits (i.e., not in $\mathsf{P/poly}$ for the decision variant of the problem). Time-evolution scenarios (where $U(x)$ represents the time-evolution of a time-dependent or independent Hamiltonian) and ground-state scenarios (where $U(x)\ket{0}$ is the desired ground state) are natural candidates for such state preparation. To provide the strongest arguments that the resulting map is also classically hard, it is useful to work with Hamiltonians known to be computationally hard (i.e. $\BQP$- or $\mathsf{QMA}$- hard in the context of time-evolution or ground-state problems, respectively). With these conditions in mind, we can still identify numerous previously recognized families of physical Hamiltonians that satisfy these criteria. For time-evolution, to list a few known cases of $\BQP$-complete Hamiltonians and otherwise specified evolutions which give raise to $\BQP$-complete problems: 
    \begin{itemize}
        \item[\footnotesize$\bullet$] Various variants of the Bose-Hubbard model~\cite{childs2013universal}.
        \item[\footnotesize$\bullet$]  Stoquastic Hamiltonians~\cite{janzing2006bqp}.
        \item[\footnotesize$\bullet$] Ferromagnetic Heisenberg model~\cite{childs2013universal}.
        \item[\footnotesize$\bullet$] The $XY$ model~\cite{piddock2015complexity}.
        \item[\footnotesize$\bullet$] Estimating the scattering probabilities in massive field theories~\cite{jordan2018bqp}, (more precisely, estimating the vacuum-to-vacuum transition amplitude, in the presence of spacetime-dependent classical sources, for a massive scalar field theory in $(1+1)$ dimensions).
        \item[$\bullet$] Simulation of topological quantum field theories~\cite{freedman2002simulation}.
    \end{itemize}
     For the ground state version, $\mathsf{QMA}$ hardness results are known for the following Hamiltonians (the first two examples are actually in $\mathsf{MA}$, a subset of $\mathsf{QMA}$): 
     \begin{itemize}
     \item[\footnotesize$\bullet$] Stoquastic local Hamiltonians~\cite{bravyi2006complexity}.
     \item[\footnotesize$\bullet$] Bose-Hubbard with negative hoppings~\cite{bravyi2006complexity}.
     \item[\footnotesize$\bullet$] Bose-Hubbard with poistive hoppings~\cite{childs2014bose}.
     \item[\footnotesize$\bullet$] Antiferromagnetic Heisenberg model~\cite{piddock2015complexity}.
     \item[\footnotesize$\bullet$] The $XY$ model~\cite{piddock2015complexity}.
     \item[\footnotesize$\bullet$] Electronic structure problem~\cite{o2021electronic} for quantum chemistry.
    \item[\footnotesize$\bullet$] Certain supersymmetric quantum mechanical theories~\cite{crichigno2024clique}.
      \end{itemize}
     
   We note that the list above is, to our knowledge, not exhaustive. Focusing on the ground state setting, the computational hardness involved in preparing these states is reflected in our quantum learning algorithm, however, we do not imply that either nature (in the realization of the dataset) or our quantum algorithm is solving $\mathsf{QMA}$-hard problems. In practical scenarios, any $\mathsf{QMA}$ problem can include instances that are solvable within $\mathsf{BQP}$ (as indeed $\BQP$ is contained in $\mathsf{QMA}$, and so all $\BQP$ problems correspond to some instances of $\mathsf{QMA}$-hard problems). For those distributions over the inputs, which are arguably the ones realized in nature, there can be efficient quantum algorithms. Specifically, while the Hamiltonian families listed above define a broader class of physically relevant systems, the instances encountered in real-world scenarios typically fall within easier subclasses that are likely to involve $\mathsf{BQP}$-hard problems. This aligns well with our framework. An example of a setting where the quantum learning algorithm is easy to specify occurs in systems where classical approximations of ground states (such as mean-field solutions, matrix product approximations, etc.) are not exponentially far from the true ground states. These systems, while belonging to $\mathsf{QMA}$-like families, could result in $\BQP$-complete problems (by virtue of the guided Hamiltonian problem construction~\cite{gharibian2022improved}). Other examples of quantumly efficient simulable processes involve Hamiltonians connected by an inverse-polynomially gapped adiabatic path, which can occur in critical Hamiltonians or near phase transition boundaries, with the appropriate scaling of system size.
   In all of these cases, our approach not only provides a natural learning scenario but also an efficient quantum learning algorithm that effectively leverages a guiding state. Finally, as we transition to practical applications, it is important to note that the emphasis on $\BQP$-complete problems is primarily relevant when aiming for formal proofs of asymptotic advantages. In particular, our separation results holds for any quantum process which maps $x\rightarrow\rho(x)$ and for which there exists an observable $O$ such that the function: $x\rightarrow \Tr[\rho(x)O]$ is outside $\mathsf{P\slash poly}$, thus contained in the class $(\mathsf{P\slash poly})^{\mathrm{c}}\cap\BQP$~\footnote{$(\mathsf{P\slash poly})^{\mathrm{c}}$ indicates the complement of the class $\mathsf{P\slash poly}$.} ($\subseteq \BQP$-complete). Many systems of interest may not be $\BQP$-complete but are likely beyond the reach of classical poly-sized circuits. Potential examples may include problems in glassy systems, Hamiltonians related to exotic highly correlated states of matter (e.g., spin liquids), highly correlated molecules (such as those involving heavier metals), and certain areas of high-energy physics, like quantum chromodynamics. While many of these problems may not be strictly $\BQP$-complete, they may remain intractable for classical computers and could potentially be addressed by quantum ones, making them strong candidates for demonstrating quantum advantages in our learning task. We acknowledge, however, that the hardness of these last examples listed above remains purely conjectural and has yet to be more precisely assessed.

\item \textbf{Data distribution: heralding without full control.} 
In the formulation of our learning problems, we assume that the input Hamiltonians, or more generally the resulting states, are drawn from a distribution. While it may not be unusual to have no information about which Hamiltonian is acting on the system, our framework assumes that this information is also revealed in the data point $x$.  Specifically, for our approach to work, we require both that the Hamiltonian is chosen randomly and that we are later informed about which Hamiltonian was selected—essentially heralding without control~\footnote{If the Hamiltonian were fully controlled, we would no longer be in a PAC learning scenario, but something closer to Angluin's query model~\cite{angluin1988queries}. While this setup would also be interesting, the learning objectives would differ.}.
Such scenarios can occur in situations where control is limited, but some type of partial (process) tomography allows the reconstruction of the implemented Hamiltonian afterward. For example, in fields like materials science and condensed matter physics, we frequently do not have full control over the conditions affecting a sample, but we can analyze it afterward. This is particularly true in cases like material doping or the introduction of nitrogen-vacancy (NV) centers. In condensed matter systems, doping introduces additional electrons or holes into a material. However, the precise concentration and spatial distribution of dopants are challenging to control at the atomic scale, leading to variability in the resulting electronic structure~\cite{stavrinadis2016strategies}. Similarly, lattice imperfections, such as strain, dislocations, or defects, perturb key Hamiltonian parameters like hopping amplitudes and interaction strengths, often in unpredictable ways~\cite{faghihnasiri2020study}. Impurities, including unintentional substitutions or vacancies, modify on-site energies and introduce disorder~\cite{khosravian2022impurity}. Collectively, these factors result in effective Hamiltonians that deviate from the intended design, necessitating a statistical approach to account for variability. Such challenges are not unique to condensed matter systems. In quantum optics, heralded operations used to generate entangled photon pairs rely on probabilistic detection events, meaning the effective realized quantum operation depends on external factors such as photon loss or background noise. Similarly, in high-energy physics, the governing Hamiltonian of particle collisions is subject to environmental factors like background radiation or fluctuating particle flux, making precise control impractical. In both cases, experimentalists work within a probabilistic framework, relying on statistical inference rather than attempting direct control of the Hamiltonian. Our model reflects these realities by treating input Hamiltonians as drawn from an underlying distribution. This approach is well-suited to systems with intrinsic or practical variability, where repeated experiments often produce diverse realizations of effective Hamiltonians. By adopting a distributional perspective, we capture the statistical nature of Hamiltonian variability and address learning problems under realistic experimental constraints. Furthermore, since our work requires learnability for any input distribution, rather than focusing on specific distributions as in~\cite{gyurik2023exponential}, we more accurately address the case where the underlying distributions are unknown. This makes our framework further aligned with real-world settings than previous approaches.

\item \textbf{Learning unknown observables.} 
   There are many natural scenarios where the actual measurement performed is not fully characterized or is entirely unknown. One clear example is learning order parameters (as observables) based on, e.g., data labeled using indirect indicators, such as entanglement entropy, which are not simple linear functions of the state. A similar scenario occurs when the order parameter is known, but the corresponding observable cannot be directly implemented on the given device. In these cases, a learning model can infer a proxy for the true order parameter as a linear combination of observables that are feasible to implement.
   Additionally, since time-evolution under a parametrized Hamiltonian can be considered part of the measurement in the case of unitarily parametrized observables, the scenario of partially unknown observables also includes situations where certain aspects of the dynamics are unknown, while the measurement itself remains fixed. Another relevant setting concern experimental contexts where, due to incomplete device characterization, the measurements realized may not align with the intended ones, reflecting a partially unknown scenario. It is important to note that limiting the degree of ignorance in the setting corresponds to restricting or biasing the concept class structure. In this case, learnability still holds (regardless of the prior) along with classical hardness, as long as there remains at least one ``hard'' measurement within the concept class.
   The characterization and learning of observables have been extensively studied, precisely because it represents a plausible scenario, especially within the framework of quantum measurement tomography, which was first introduced in~\cite{luis1999complete}. Specifically, in~\cite{luis1999complete} a setting where unknown observables emerge was identified, when a system interacts with external degrees of freedom before being measured by a fixed observable. This is particularly relevant in cases involving reservoirs or other mechanisms that induce losses and decoherence effects. Subsequent works, such as~\cite{lundeen2009tomography}, used the task of learning unknown measurements to characterize experimental detectors. This task is especially pertinent in quantum computing, where measurement devices are often imperfect due to noise, making it crucial to understand and mitigate the effects of these imperfections on quantum measurements.

\end{enumerate}
Combining meaningful settings from these three classes of assumptions results in more realistic learning tasks that can be captured by our abstract learning framework.

\section*{Acknowledgments}
This publication is part of the project Divide \& Quantum (with project number 1389.20.241) of the research
programme NWA-ORC which is (partly) financed by the Dutch Research Council (NWO). This work was also supported by the Dutch National Growth Fund (NGF), as part of the Quantum Delta NL programme. This work was also supported by the European Union’s Horizon Europe program through the ERC CoG BeMAIQuantum (Grant No. 101124342)

\appendix

\section{Proof of classical hardness}\label{app: classical hardness}
As anticipated in the main text, in order to prove our main Theorem~\ref{thm: separation} we need to show both classical hardness and quantum learnability for the time-evolution learning problem. In this section, we prove Theorem~\ref{lemma: cl hardness} for the time-evolution learning problem, formally described through the concept class $\mathcal{F}^{H,O}_{\mathrm{evolved}}$.
It is important to notice that, as we discussed in the main text, the same proofs can be easily extended for the class $\mathcal{F}^{\mathcal{H},O}_{\mathrm{g. s.}}$ of measurements on ground state of local Hamiltonians through the Kitaev's circuit-to-hamiltonian construction.

\subsection{$\BQP$-completeness of constant time-evolution}

Before stating the proof for Therem~\ref{lemma: cl hardness}, let us prove an useful lemma which guarantees the $\BQP$-completeness of constant time Hamiltonian evolution following the idea in present in~\cite{childs2004quantum}. It is important to notice that the Hamiltonians used in Lemma~\ref{lemma: constant time ev} will not have a constant norm. This is however not a problem for our quantum algorithm as their norms will scale only polynomially with the number of qubits.

\begin{lemma}[Constant time-evolution is $\BQP$-complete]\label{lemma: constant time ev}
For any $k-$gate quantum circuit $U=U_k...U_2U_1$ which acts on $n$ qubits there exists a local Hamiltonian $H$ such that for any $n$ qubit inital state $\ket{\psi}$ :
\begin{equation}\label{eq: perfect evolution}
    e^{iHt}\ket{\psi}\ket{0}=U_k...U_2U_1\ket{\psi}\ket{k}
\end{equation}
for $t=\pi$.

\begin{proof}
    Consider the Feyman clock Hamiltonian \cite{feynman1985quantum,nagaj2010fast}:
    \begin{equation}\label{eq: feynmann ham}
     H:=\sum_{j=1}^k H_j
    \end{equation}
    where:
    \begin{equation}
        H_j= U_j\otimes\ket{j}\bra{j-1}+U^\dagger_j\otimes\ket{j-1}\bra{j}
    \end{equation}
notice that using the common unary encoding for the clock register, such Hamiltonian is 4 local \cite{nagaj2010fast}. 
Furthermore, the Hamiltonian $H$ acts on two different registers: the first register (``work register'') consists of $n$ qubit and it will store the computation of the circuit $U$, the second register (``clock register'') contains $k+1$ qubit and it acts as a counter which records the progress of the computation. Now, if we evolve an initial state $\ket{\psi_0}=\ket{\psi}\ket{0}$ under the Hamiltonian $H$ the evolved state will be in the space spanned by the $k+1$ states $\{\ket{\psi_j}=U_j...U_2U_1\ket{\psi}\ket{j}\}_{j=0}^{k}$.
After letting the system evolve for a time $\tau$, if now we measure the clock register and obtain a value $L$ then the work register will exactly contains the computation of the quanutm circuit $U$ on the initial state $\ket{\psi}$. This of course will happen only with a certain probability and while there are many ways to boost such probabilities of getting the desired final state $\ket{\psi_k}=U_k...U_2U_1\ket{\psi}\ket{k}$, the idea in~\cite{childs2004quantum} is to modify the Feynman's Hamiltonian in Eq.~(\ref{eq: feynmann ham}) in order to make the evolution perfect as in Eq.~(\ref{eq: perfect evolution}).

Notice first that in the subspace spanned by the vectors $\{\ket{\psi_j}\}_{j=0}^{k}$ the non zero entries of the matrix $H$ are 
\begin{equation}\label{eq: H matrix}
    \bra{\psi_j}H\ket{\psi_{j\pm 1}}=1
\end{equation}
We now follow the idea in~\cite{childs2004quantum} and modify the Hamiltonian in~(\ref{eq: feynmann ham}) in the following way:
\begin{equation}
    H':=\sum_{j=1}^k \sqrt{j(k+1-j)}\;H_j
\end{equation}
The idea in~\cite{childs2004quantum} is to associate each state in $\{\ket{\psi_j}\}_{j=0}^k$ to a quantum system of total angular momentum $\frac{k}{2}(\frac{k}{2}-1)$ with $z$ component $j-\frac{k}{2}$. The association is possible by the fact that a system with total angular momentum $\frac{k}{2}(\frac{k}{2}-1)$ will allow $k+1$ states with $z$ components of values $-\frac{k}{2}, -\frac{k}{2}+1,...,\frac{k}{2}-1,\frac{k}{2}$. These are exactly the $k+1$ states in $\{\ket{\psi_j}\}_{j=0}^k$ and it is possible to move among them defining the corresponding ladder operators :
\begin{align}\label{eq: ladders}
    L_-\ket{\psi_j} &= \sqrt{j(k+1-j)}\ket{\psi_{j-1}}\\
    L_+\ket{\psi_j} &= \sqrt{(k-j)(j+1)}\ket{\psi_{j+1}}
\end{align}
By the algebra of the angular momentum, the $x$ component of the total angular momentum $J_x$ can be expressed as
\begin{equation}
    J_x=\frac{1}{2}(L_+ + L_-)
\end{equation}
Comparing Eq.(\ref{eq: H matrix}) with Eq. (\ref{eq: ladders}) it is clear that $H'$ is exactly the $x$ component of the angular momentum operator defined over the states $\{\ket{\psi_j}\}_{j=0}^k$. As $J_x$ rotates between the states with $z$ component $\pm \frac{k}{2}$ in time $t=\pi$, this concludes the proof.

\end{proof}
    
\end{lemma}
\subsection{Proof of Theorem~\ref{lemma: cl hardness}}

The second ingredient in demonstrating the classical hardness of our learning task, in addition to the $\BQP$-completeness of time-evolution, is a result from classical learning theory contained in~\cite{schapire1990strength}. Here, we will state this result in Lemma~\ref{lemma: shapire} and, for clarity and didactical purpose, provide an intuitive explanation of its proof.

\begin{lemma}[Learnability implies evaluation - Theorem 7 in ~\cite{schapire1990strength}]\label{lemma: shapire}
    Suppose $\calF$ is learnable in the sense of Def.~\ref{def:pac} with input space $\calX_n=\{0,1\}^n$. Then there exists a polynomial $p$ such that for all the concepts $f\in\calF$ of size $s$, there exists a classical circuit of size $p(n,s)$ exactly computing $f$ 
\end{lemma}
\begin{proof}[Proof sketch]
    In Lemma~\ref{lemma: shapire}, the size $s$ of a concept refers to the length of its representation under some encoding scheme. For example, in the case of our concept class $\calF_{\mathrm{evolved}}^{H_{\mathrm{hard}},O}$, $s\sim\mathrm{poly}(n)$ as our concepts can be implemented by polynomial size quantum circuits. The key result underlying Lemma~\ref{lemma: shapire} is another finding from~\cite{schapire1990strength}, which states the following. If the concept class $\calF$ is learnable in the sense of Def.~\ref{def:pac}, then there exists a learning algorithm that, for any target concept $f\in\calF$, produces a hypothesis implementable by a circuit of size $\mathrm{poly}(n,s,\log(1/\epsilon))$, which is $\epsilon$-close to the target concept in the sense of Eq.\ref{eq:real pac}. This remarkable result is achieved through a smart application of boosting theory. The proof of Lemma~\ref{lemma: shapire} then follows by applying the boosted learning algorithm to a training set containing all input points $x\in\{0,1\}^n$ and requiring an exponential precision $\epsilon\leq\frac{1}{2^n}$.
\end{proof}

We are now ready to state the proof of Theorem~\ref{lemma: cl hardness} on the existence of a Hamiltonian $H_{\mathrm{hard}}$ such that the concept class $\calF_{\mathrm{evolved}}^{H_{\mathrm{hard}},O}$ is not classically learnable in the sense of Def.~\ref{def:learning}.

\hardness*
\begin{proof}
   Let $\{U^n_{BQP}\}_n$ be a family of quantum circuits which decides an arbitrary $\BQP$ language $\mathcal{L}$, one circuit per size. Precisely, according to the definition of $\BQP$ in Def.~\ref{def:BQP}, this implies that for any $\bm{x}\in\{0,1\}^n$ measuring the first qubit in the computational basis on the state $U^n_{BQP}\ket{\bm{x}}$ will output 1 (or 0) with probability greater that 2/3 if $\bm{x}\in\mathcal{L}$ ($\bm{x}\not\in\mathcal{L}$). It is obvious that the the quantum model $f^{O'}=\text{Tr}[\rho(\bm{x})O']$ correctly decides every $\bm{x}\in\mathcal{L}$ if $\rho(\bm{x})=U^n_{BQP}\ket{\bm{x}}\bra{\bm{x}}(U^n_{BQP})^\dagger$ and $O'=Z\otimes \underbrace{I\otimes ...\otimes I}_{n-1}$. In fact measuring the first qubit of $U^n_{BQP}\ket{\bm{x}}$:
    \begin{enumerate}
    \item For all $\bm{x}\in \mathcal{L}$, $f^{O'}(\bm{x})=$ Pr[the output of $U^n_{BQP}$ applied on the input $\ket{\bm{x}}$ is $1$] - Pr[the output of $U^n_{BQP}$ applied to the input $\ket{\bm{x}}$ is $0$] $\geq 2/3-1/3=1/3$.
    \item For all $\bm{x}\notin \mathcal{L}$, $f^{O'}(\bm{x})=$ Pr[the output of $U^n_{BQP}$ applied on the input $\ket{\bm{x}}$ is $1$] - Pr[the output of $U^n_{BQP}$ applied to the input $\ket{\bm{x}}$ is $0$] $\leq 1/3-2/3=-1/3$.
\end{enumerate}
  Therefore, as $f^{O'}(\bm{x})>0$ if $\bm{x}\in \mathcal{L}$ and $f^{O'}(\bm{x})<0$ if $\bm{x}\notin \mathcal{L}$ such quantum model could efficiently decide the language $\mathcal{L}$.
 By Lemma~\ref{lemma: constant time ev}, for every $n$ there exist a 4-local Hamiltonian $H_{\mathrm{\mathrm{hard}}}$ such that evolving the state $\ket{\bm{x}}$ for a time $t=\pi$ under will produce the state $U^n_{\BQP}\ket{\bm{x}}$ (on the work register). 
  Therefore, the concept class $\calF_{\mathrm{evolved}}^{H_{\mathrm{hard}},O}$ contains a concept $f^{\bm{\alpha'}}$, where $\bm{\alpha'}$ is defined such that $O(\bm{\alpha'}) = O'$, which correctly decides $\calL$. The final step of the proof lies in noticing that by Lemma~\ref{lemma: shapire}, if $\calF_{\mathrm{evolved}}^{H_{\mathrm{hard}},O}$ is learnable in the sense of Def.~\ref{def:learning}, then there must exists a polynomial size classical circuit which evaluates $f^{\bm{\alpha'}}$ correctly on every $\bm{x}\in\{0,1\}^n$. As $f^{\bm{\alpha}}$ decides the $\BQP$ language $\calL$, and for any $\BQP$ language we can construct such Hamiltonian $H_{\mathrm{hard}}$, this implies that $\BQP\subseteq\mathsf{P\slash poly}$.

\end{proof}

\section{Proof of quantum learnability}\label{app: quantum learnability}
In this Section we prove that Algorithm~\ref{algo:quantum} satisfies the learning condition in Def.~\ref{def:learning} using only polynomial resources in both time and samples. As a first step, notice that the quanutm states $\rho_{H_{\mathrm{hard}}}(\bm{x})$ in Eq. (\ref{eq:concept class}) of $\mathcal{F}^{H_{\mathrm{hard}},O}_{\mathrm{evolved}}$ are easily preparable on a quantum computer. This is due to the fact that the Hamiltonians $H_{\mathrm{hard}}$ used to establish our hardness result in Theorem~\ref{lemma: cl hardness} are local, allowing their time-evolution to be efficiently simulated on a quantum computer.
Since the quantum algorithm can efficiently prepare time evolved states of local Hamiltonians, the quantum states $\rho_{H_{\mathrm{hard}}}(\bm{x})=U\ket{\bm{x}}\bra{\bm{x}}U^\dagger$ with $U=e^{iH_{\mathrm{\mathrm{hard}}}\pi}$ can efficiently be prepared on a quantum computer. In order to demonstrate the quantum learnability of $\mathcal{F}^{H_{\mathrm{hard}},O}_{\mathrm{evolved}}$, all that remains is to rigorously bound the sample and time complexity of Algorithm $\ref{algo:quantum}$ in the following Theorem \ref{thm:quantum algo bounds}. As the bound obtained on the sample complexity derives from the generalization bound of the LASSO algorithm \cite{mohri2018foundations,lewis2024improved}, we first repeat this result.

\begin{theorem}[Theorem 11.16 in \cite{mohri2018foundations}]\label{thm:lasso}
    Let $\mathcal{X} \subseteq \mathbb{R}^{A}$ and $\mathcal{C} = \{\bm{x}\in \mathcal{X} \mapsto \vec{w} \cdot \bm{x}: \|\vec{w}\|_1 \leq B\}$. Let $\mathcal{S} = ((\bm{x}_1, y_1), \dots, (\bm{x}_N, y_N)) \in (\mathcal{X} \times \mathcal{Y})^N$. Let $\mathcal{D}$ denote a distribution over $\mathcal{X} \times \mathcal{Y}$ according to which the training data $\mathcal{S}$ is drawn. Assume that there exists $r_{\infty} > 0$ such that for all $\bm{x}\in \mathcal{X}$, $\|\bm{x}\|_{\infty} \leq r_{\infty}$ and $M > 0$ such that $|h(\bm{x}) - y| \leq M$ for all $(\bm{x}, y) \in \mathcal{X} \times \mathcal{Y}$. Then, for any $\delta > 0$, with probability at least $1 - \delta$, each of the following inequalities holds for all $h \in \mathcal{C}$:
\begin{align}
&\mathbb{E}_{(\bm{x},y) \sim \mathcal{D}} [|h(\bm{x}) - y|^2] := R(h) \\ &\leq \hat{R}_{\mathcal{S}}(h) + 2r_{\infty}B M\sqrt{2 \log(2|A|) N} + M \sqrt{2\log(1/\delta)/{2N}} \notag
\end{align}
where $R(h)$ is the prediction error for the hypothesis $h$ and $\hat{R}_{\mathcal{S}}(h)$ is the training error of $h$ on the training set $\mathcal{S}$.

\end{theorem}

We are now ready to state Theorem \ref{thm:quantum algo bounds} (a rigorous formulation of Lemma~\ref{lemma: q easiness} in the main text) which provides precise guarantees on the number of samples and the time complexity required by Algorithm \ref{algo:quantum}.
\begin{theorem}[Lemma~\ref{lemma: q easiness} in the main text]\label{thm:quantum algo bounds}
    Given $n$, $\delta> 0$, $ \frac{1}{e} >\epsilon > 0$ and a training data set $\mathcal{T}^{\bm{\alpha}}_N=\{(\bm{x}_\ell,y_\ell)\}_{\ell=1}^N$ of size
    \begin{equation}
        N=\mathcal{O}\left(\frac{\log(\mathrm{poly}(n)/\delta)\mathrm{poly}(n)}{\epsilon^2}\right)
    \end{equation}
    where $\bm{x}_\ell$ is sampled from an unknown distribution $\mathcal{D}$ over $\bm{x}\in\{0,1\}^n$ and $|y_\ell-\text{Tr}(O({\bm{\alpha}})\rho(\bm{x}_\ell))|\leq \epsilon$ for any geometrically local observable $O({\bm{\alpha}})=\sum_{i=1}^m {\bm{\alpha}}_i P_i$ with $P_i \in \{I,X,Y,Z\}^{\otimes n}$ and ${\bm{\alpha}}\in [-1,1]^m$, and $\rho(\bm{x})=U(\bm{x})\ket{0}\bra{0}U^\dagger(\bm{x})$. Then there exists a quantum algorithm $\mathcal{A}_q(\mathcal{T}^\alpha_N,\bm{x})=h(\bm{x})$ such that:
    \begin{equation}
        \mathbb{E}_{x\sim\mathcal{D}}|h(\bm{x}_l)-\text{Tr}[O({\bm{\alpha}})\rho(\bm{x})]|^2\leq\epsilon
    \end{equation}
    with probability at least $1-\delta$. The computational time of the quantum algorithm is bounded in $\mathcal{O}(\mathrm{poly}(n)N)$
\end{theorem}

\begin{proof}
    The proof of the theorem is based on the well known bound on the prediction error of the LASSO algorithm of Theorem \ref{thm:lasso}. Consider the algorithm described in Algorithm \ref{algo:quantum}. First, we demonstrate now that Algorithm \ref{algo:quantum} satisfies the condition of Theorem 6. In our setting, the input space of the learned model $h$ is $\mathcal{X}=[-1,1]^m\subset\mathbb{R}^m$ as we consider $h$ a function of the $m-$dimensional feature vector $\phi(\bm{x})$. Clearly, $||\phi(\bm{x})||_{\infty}\leq 1$ for all $\bm{x}\in\mathcal{X}$. The hypothesis class $\mathcal{C}$ of our algorithm is given by the set of functions of the same form of the learned $h$, i.e. $\mathcal{C}=\{\phi(\bm{x})\in\mathcal{X} \rightarrow w\cdot \phi(\bm{x}) : ||w||_1\leq B\}$ with $B=\mathrm{poly}(n)$. With respect to Theorem 6, we can also choose $M= \mathrm{poly}(n)$ so that $|h(\bm{x}_l)-y_l|<M$ for all $l=1, ..., N$:
    \begin{align}
        |h(\bm{x}_l)-y_l|&\leq|w\cdot\phi(\bm{x}_l)|+|y_l|\leq||w||_1||\phi(\bm{x}_l)||_{\infty}+2\\&\leq \text{poly}(n)\notag
\end{align}
Where the second inequalitiy follows by H\"{o}lder's inequality.
By Theorem~\ref{thm:lasso} then, the bound on the prediction error $R(h)$ of the learned model $h(\bm{x})=w^*\phi(\bm{x})$ is
\begin{equation}\label{gen_error}
    R(h)\leq \hat{R}(h) + 2BM\sqrt{\frac{2\log(2m)}{N}}+M^2\sqrt{\frac{\log\delta^{-1}}{2N}}
\end{equation}
where $\hat{R}(h)$ is the training error on the dataset $T^{\bm{\alpha}}_N$.

We now bound the training error $\hat{R}(h)$. Let $\epsilon_1$ be the maximum sampling error associated to the measurement of a Pauli observable on $\rho(\bm{x}_\ell)$, i.e.
\begin{equation}
    \epsilon_1 = \max_{\substack{i\in[1,...,m]\\ \bm{x}\in\{0,1\}^n}}\;\;|\mathrm{Tr}[\rho(\bm{x})P_i]-\phi(\bm{x}_i)|
\end{equation}
Let also $\epsilon_2$ the maximum sampling error associated to $y_\ell$, such that:
    
    \begin{equation}
        |y_\ell - \text{Tr}[O({\bm{\alpha}})\rho(\bm{x}_l)]|\leq \epsilon_2\;\;\;\;\; \forall (\bm{x}_\ell,y_\ell)\in T^{\bm{\alpha}}_N
    \end{equation}

We can now derive a bound on the training error for the optimal value of $w_{opt}={\bm{\alpha}}$ of the model $h_{opt}(\bm{x})={\bm{\alpha}}\cdot\phi(\bm{x})$:
\begin{align}
    \hat{R}(h_{opt})&=\frac{1}{N}\sum_{\ell=1}^N |h_{opt}(\bm{x}_\ell)-y_\ell|^2 \leq \max_\ell |h(\bm{x}_\ell)-y_\ell| \\ 
    &\leq (|\text{Tr}[\rho(\bm{x}_{\ell^*})O({\bm{\alpha}})]-h(\bm{x}_{\ell^*})|\;+\\
    &\;\;\;\;\;\;|\text{Tr}[\rho(\bm{x}_{\ell^*})O({\bm{\alpha}})]-y_\ell|)^2 \notag\\
    &\leq \left(\left( \sum_i |{\bm{\alpha}}_i|\right)\epsilon_1 + \epsilon_2\right)^2
    \\
    &\leq (B\epsilon_1+\epsilon_2)^2
\end{align}
In practice, we can require to the LASSO algorithm to obtain a $w^*$ which achieves a training error at most $\epsilon_3/2$ larger than the optimal one. 
We can obtain such precision by setting $B=\mathrm{poly}(n)$. Formally, we have 
\begin{equation}
    \hat{R}(h)\leq \frac{\epsilon_3}{2}+\min_{\substack{w\in\mathbb{R}^m\\ ||w||_1\leq B}} 
    \frac{1}{N}\sum_{\ell=1}^N |w\cdot\phi(\bm{x}_\ell)-\text{Tr}[\rho(\bm{x}_\ell)O({\bm{\alpha}})]|^2
\end{equation}
Because we have set $B=\mathrm{poly}(n)$, we have that the second term must be at most $\hat{R}(h_{opt})$ and therefore we have
\begin{equation}
    \hat{R}(h)\leq(B\epsilon_1+\epsilon_2)^2+\frac{\epsilon_3}{2} 
\end{equation}
We note that as $B\leq\mathcal{O}(\mathrm{poly}(n))$ we can bound the error $\epsilon'_1=B\epsilon_1$ to scale polynomially to zero by just reducing the sampling error $\epsilon_1$ using polynomially many more copies of each $\rho(\bm{x}_l)$.
Thus we can rewrite equation (\ref{gen_error}) as
\begin{equation}
    R(h)\leq(\epsilon'_1+\epsilon_2)^2+\frac{\epsilon_3}{2}+2BM\sqrt{\frac{2\log(2m)}{N}}+M^2\sqrt{\frac{\log(\delta)^{-1}}{2N}}
\end{equation}
Then, in order to bound the prediction error above by $\epsilon=(\epsilon'_1+\epsilon_2)^2+\epsilon_3$ we need to choose an N such that
\begin{equation}\label{eq:eps_3/2}
    2BM\sqrt{\frac{2\log(2m)}{N}}+M^2\sqrt{\frac{\log(\delta)^{-1}}{2N}} \leq \frac{\epsilon_3}{2}
\end{equation}
By substituting $M=B$ we obtain that:
\begin{align}
    &2BM\sqrt{\frac{2\log(2m)}{N}}+M^2\sqrt{\frac{\log(\delta)^{-1}}{2N}} \notag\\&\leq \left(\frac{B^2}{\sqrt{2N}}(4\sqrt{\log(2\mathcal{O}(\text{poly}(n)))}+\sqrt{\log(\delta^{-1})} \right)
\end{align}
so that it is upper bounded by $\epsilon_3/2$ choosing $N$ as
\begin{align}
    N&=2\frac{B^4\sqrt{2\log(\mathcal{O}(\text{poly}(n))/\delta)}}{(\epsilon_3)^2}\\
    &\sim\frac{\log(\text{poly}(n)/\delta)\text{poly}(n)}{(\epsilon_3)^2}
\end{align}

By setting $\epsilon'_1=0.2\epsilon$, $\epsilon_2=\epsilon$, and $\epsilon_3=0.4\epsilon$ we have $(\epsilon'_1+\epsilon_2)^2+\epsilon_3\leq\epsilon$ and recover the claim of the Theorem.

Finally we bound the efficiency of our algorithm. The training time is dominated by the creation of the feature map $\phi(\bm{x}_\ell)$ for each training point and by the LASSO regression over the corresponding feature space. To create the vector $\phi(\bm{x})$ the quantum algorithm needs to prepare multiple copies of $\rho(\bm{x}_\ell)$ $\forall \bm{x}_\ell \in T_N^{\bm{\alpha}}$. As seen before, only a polynomial number of copies are sufficient to achieve a desired error $\epsilon'_1$, so that the whole process takes time $\mathcal{O}(\mathrm{poly}(n)N)$.
For the LASSO regression, it is known that to obtain a training error at most $\epsilon_3/2$ larger than the optimal function value, the LASSO algorithm on the feature space of $\phi(\bm{x})$ can be executed in time $\mathcal{O}\left(\frac{m_\phi\log m_\phi}{\epsilon^2_3}\right)$ \cite{hazan2012linear}, where in our case $m_\phi=m$. It is easy to show that even this time is bounded by $\mathcal{O}(\mathrm{poly}(n)N)$. 
The prediction time corresponds to the time for the evaluation of the learned model $h(\bm{x})=w^*\cdot\phi(\bm{x})$ which takes time $\mathcal{O}(m)\sim\mathcal{O}(\mathrm{poly}(n))$.
In conclusion, the overall time of the quantum algorithm is bounded by $\mathcal{O}(\mathrm{poly}(n)N)$

\end{proof}

\section{Construction of the Kitaev's circuit-to-Hamiltonian for the ground state }
To make our work more self-contained, we report the construction of the Kitaev's circuit-to-Hamiltonian which was used in the proof of Theorem~\ref{thm: separation gs}. We follow the constructions outlined in~\cite{kitaev2002classical,kempe20033}. Let $x\in\mathcal{L}$ an input of a language $\calL\in\BQP$ and $U(x)=U_T...U_1$ be a quantum circuit of polynomial size $T$ acting on $n=\mathrm{poly}(|x|)$ qubits which decides $x$. We first add $2T$ layers of identities to $U(x)$.
The Hamiltonian \( H \) that is constructed operates on a space of \( n = N + \log(3T + 1) \) qubits. The first \( N \) qubits represent the computation, and the last \( \log(3T + 1) \) qubits represent the possible values \( 0, \dots, 3T \) for the clock. The Hamiltonian is constructed from three terms:

\begin{align}
\label{eq:kitaev}
H(x) = H_{\mathrm{init}} + H_{\mathrm{clock}} + \sum_{t = 1}^{3T} H_t(x),
\end{align}
with
\begin{align}
\centering
H_{\mathrm{init}} &=  \sum_{i = 1}^N \ket{0}\bra{0}_i,\quad
H_{\mathrm{clock}} =  \sum_{t = 1}^{3T -1}\ket{01}\bra{01}^{\mathrm{clock}}_{t, t+1},\\
H_t(x) = \frac{1}{2}\Big(&I \otimes \ket{100}\bra{100}^{\mathrm{clock}}_{t-1, t, t+1} \\ 
       &+ I \otimes \ket{110}\bra{110}^{\mathrm{clock}}_{t-1, t, t+1} \\
&- U_t(x) \otimes \ket{110}\bra{100}^{\mathrm{clock}}_{t-1, t, t+1} \\
&- U_t(x)^\dagger \otimes \ket{100}\bra{110}^{\mathrm{clock}}_{t-1, t, t+1}\Big)
\end{align}
where $\ket{.}\bra{.}_i$ acts on the $i$th site of $\mathbb{C}^{2^N}$, $\ket{.}\bra{.}_j^{\mathrm{clock}}$ acts on the $j$th site of $\mathbb{C}^{2^{3T}}$ and $U_t$ denotes the $t$th layer of gates in $U(x)$. 
Note that $H(x)$ is 5-local for all $x $.
The ground state of $H(x)$ is given by $\rho(x) = \ket{\psi(x)}\bra{\psi(x)}$, where
\begin{align}
\ket{\psi(x)} = \frac{1}{\sqrt{3T}} \sum_{t = 1}^{3T}(U_t\cdots U_1)(x)\ket{0^N}\ket{1^t0^{3T - t}},
\end{align}
Measuring the local observable $O=Z\otimes I\otimes \ket{1}\bra{1}_T^{\mathrm{clock}}$ on the state $\ket{\psi(x)}$ will decide the input $x$, analogously to the case in the proof of Theorem~\ref{lemma: cl hardness}.

\section{Proof of Theorem~\ref{thm: separation W} and Corollary~\ref{lemma: learning unitary}}\label{app:unitaries}

We rewrite here Theorem~\ref{thm: separation W} form the main text, formalizing the result and providing a detailed proof.

\begin{theorem}[Theorem~\ref{thm: separation W} in the main text]
Assuming that $\BQP\not\subseteq\mathsf{P\slash poly}$, the following holds. Let $S=\{\ket{\psi_\ell}\}_\ell$ and $Q=\{Q_m\}_m$ be a discrete set of quantum states and measurement operators, respectively, both of which are polynomially describable. Then for every efficient (non-adaptive) algorithm $\mathcal{A}_W$ which learns an arbitrary unitary $W$ from query access to it on states from $S$ and measured by operators in $Q$ there exists a set of distributions $\{\mathcal{D}_i\}_i$ over $\{0,1\}^n$, a family of unitaries $\{U(\bm{x})\}_{\bm{x}}$ and a measurement $O$ for which the concept class $\mathcal{M}_{U,W,O}$ exhibits a learning advantage under the learning condition of Def.~\ref{def:learning}.    
\end{theorem}

\begin{proof}

We can think at the algorithm $\mathcal{A}_{W}$ as an algorithm which receives in input pairs of $\{(\ket{\psi_\ell},y^{m_\ell}_\ell)_\ell\}$, where $\ket{\psi_\ell}\in S$ and $y^{m_\ell}_\ell\in\mathbb{R}$ is the measurement outcome of a randomly selected operator $Q_{m_l}\in Q$ on $\ket{\psi_\ell}$, and for any $\epsilon\geq 0 $ outputs a matrix $\Tilde{W}$ such that $||\Tilde{W}-W||\leq\epsilon$ in some norm distance. The idea to construct a $\mathcal{M}_{U,W,O}$ which exhibits a separation is the following, illustrated also in Figure \ref{fig:unitary separation}. Consider the family of quantum states $\{\ket{\psi_U(\bm{x})}\}_{\bm{x}}$ constructed by parameterized unitaries $\{U(\bm{x})\}_{\bm{x}}$ in the following way, depending on the input $\bm{x}\in\{0,1\}^n$:

\begin{itemize}

\item The first $1+n_Q$ qubit of $\ket{\psi_U(\bm{x})}$ are initializated in the state $\ket{x_1}\otimes\ket{x_2x_3...x_{n_Q+1}}$, the remaining $n_S=n-(1+n_Q)$ qubit $\bm{x_{S}}=x_{n-n_S}x_{n-(n_S-1)}...x_n$ are in the state $\ket{0^{\otimes n_S}}$.

\item If the first bit $x_1$ is 0, then $U(\bm{x})$ prepares on the $n_S$ qubits register the state $\ket{\psi_{\bm{x_S}}}$ from the set $S$ described by the bitstring $\bm{x_S}\in\{0,1\}^{n_S}$. As we assumed that every state in $S$ allows a polynomial classical description, we only need an input $\bm{x}\in\{0,1\}^n$ of polyinomial size. We thus have:
\begin{equation}
        U(\bm{x})\left(\ket{0}\otimes\ket{\bm{x_{Q}}}\otimes\ket{0^{\otimes n_S}}\right)= \ket{0}\otimes\ket{\bm{x_{Q}}}\otimes\ket{\psi_{stab}(\bm{x_S})}
\end{equation}

\item If the first bit $x_1$ is 1, then on the $n_S$ qubit register $U(\bm{x})$ prepares the state $\ket{\psi_{BQP}(\bm{x_S})}$ such that $\bra{\psi_{BQP}(\bm{x_S})}Z\otimes I \otimes ... \otimes I\ket{\psi_{BQP}(\bm{x_S})}$ outputs $+1/3$ if the $n_S$ bitstring $\bm{x_S}$ belongs to an arbitrary (previously fixed) $\BQP$complete language $\mathcal{L}$ defined over input $\bm{x_S}\in\{0,1\}^{n_S}$, while it outputs $-1/3$ otherwise. Following the same arguments used in the proof of Theorem~\ref{lemma: cl hardness}, for any $\mathcal{L}\in \BQP$ it always exists such $U(\bm{x})$. We thus have:
\begin{equation}
        U(\bm{x})\left(\ket{1}\otimes\ket{\bm{x_{Q}}}\otimes\ket{0^{\otimes n_S}}\right)= \ket{1}\otimes\ket{\bm{x_{Q}}}\otimes\ket{\psi_{BQP}(\bm{x_{S}})}
\end{equation}

\end{itemize}
Regarding the unknown observable $O({\bm{\alpha}})$, we will consider a measurement operator of the kind Eq.(\ref{eq: concept unitary}). Specifically we define a controlled operator $V_A$, controlled by the $1+n_Q$ qubit register $\ket{x_1x_2x_3...x_{n_Q+1}}$, such that $V_A=\sum_{0\bm{x_Q}}\ket{0\bm{x_Q}}\bra{0\bm{x_Q}}\otimes V(\bm{x_Q})$ with $\bm{x_Q}=x_2x_3...x_{n_Q+1}$ if $x_1=0$ and $V_A=\sum_{1\bm{x_Q}}\ket{1\bm{x_Q}}\bra{1\bm{x_Q}}\otimes I^{\otimes n_S}$ if $x_1=1$.
The unitary matrices $V(\bm{x_Q})$ are defined such that they rotate the measurement $O$ into the observable $Q_{\bm{x_Q}}\in Q$ described by the bitstring $\bm{x_Q}=x_2x_3...x_{n_Q+1}$, i.e. $Q_{\bm{x_Q}}=V(\bm{x_Q})OV^\dagger(\bm{x_Q})$ where $O$ could be taken as the local observable $Z_1=Z\otimes I^{\otimes n_S-1}$ on the $n_S$ qubit register.
The final observable $O(\bm{\alpha})$ will then be $O(\bm{\alpha})=(I^{\otimes (1+n_Q)}\otimes W(\bm{\alpha})\; V_A)O (V_A \;I^{\otimes (1+n_Q)}\otimes W(\bm{\alpha}) )^\dagger$, with $W({\bm{\alpha}})$ an arbitrary unitary parameterized by ${\bm{\alpha}}$.  Furthermore, we define $W( {\bm{\alpha}}=\Vec{0})=I\otimes ... \otimes I$. We can now construct a set of distributions $\{\mathcal{D}_i\}_i$ such that the concept class $\mathcal{M}_{U,W,O}$ exhibits a learning separation. For Theorem \ref{lemma: cl hardness}, under the learning condition in Def.~\ref{def:learning}, the concept $f^0=\text{Tr}[\rho_U(\mathbf{x})V_AOV^\dagger_A]$ cannot be efficiently learned by a classical algorithm when $U(\bm{x})$ is the circuit which decides the $\BQP$-complete language $\mathcal{L}$ over the $n_S$-sized bitstrings $\bm{x_S}$. We now define the following set distributions $\{\mathcal{D}_i\}_i$ on the input bitsrings $\bm{x}\in\{0,1\}^n$: 
\begin{itemize}
    \item The first bit $x_1$ of $\bm{x}$ is extracted randomly with equal probability between 0 or 1.
     \item If $x_1=0$ then the other $n-1$ bit $x_2x_3...x_n$ are extracted from the required distribution by $\mathcal{A}_W$ to learn the unitary $W(\bm{\alpha})$.
    \item If $x_1=1$ then the following $n_Q$ bit $\bm{x_Q}$ are extracted uniformly at random while the $n_S$ bit $\bm{x_S}$ are sampled with an arbitrary input distribution over $\bm{x_S}\in\{0,1\}^{n_S}$ which specifies the overall distribution $\calD_i$ for each $i$.
    
\end{itemize}

\textbf{Classical hardness} The classical hardness of the learning task comes directly from the proof of Theorem~\ref{lemma: cl hardness}. Consider the concept $f^0\in \mathcal{M}_{U,W,M}$ defined by $W(\bm{0})=I^{\otimes n_S}$. For the same reasoning of the proof of Theorem~\ref{lemma: cl hardness} there cannot exists a polynomial sized classical circuit which evaluates $f^0(\bm{x})$. Since for any of the input distributions $\calD_i$, $x_1$ is equally sampled between 0 and 1, no classical algorithm can meet the learning condition of Eq. \ref{def:learning} on half of the input bitstrings, thus it can not learn the concept $f^0\in\mathcal{M}_{U,W,O}$ in polynomial time for every $\epsilon$. As the learning algorithm must suceed for every $\bm{\alpha}$, this suffices to prove the classical hardness of the learning task.

\vspace{4mm}
\textbf{Quantum learnability} Recall that the data the learning algorithm receives for a concept $f^{\bm{\alpha}}\in \mathcal{M}_{U,W,M}$ are $\mathcal{T}^{\bm{\alpha}}=\{\bm{x}_\ell,y_\ell \}_\ell$ with $\mathbb{E}[y_\ell]=\text{Tr}[\rho_{U}(\bm{x}_\ell)O({\bm{\alpha}})]$ and $\rho_U(\bm{x}_\ell)=\ket{\psi_U(\bm{x})}\bra{\psi_U(\bm{x})}$. Now, in the case the first bit of $\bm{x}_\ell$ is 0, $\mathcal{T}^{\bm{\alpha}}$ are exaclty the pairs $\{(\ket{\psi_\ell},y^m_\ell)\}_\ell$ required by the algorithm $\mathcal{A}_W$ to learn $W(\bm{\alpha})$ and thus $O(\bm{\alpha})$.  
As $x_1=0$ for half of the training samples in $\mathcal{T}^{\bm{\alpha}}$, the quantum algorithm is able to learn $O(\bm{\alpha})$ and evaluate it on any input state.  

\end{proof}

\vspace{5mm}

We now provide here the the rigorous proof of Corollary \ref{lemma: learning unitary}. The proof follows the same steps of the one for Theorem~\ref{thm: separation W} above while concretizing the result for shallows $W$.

First, let us repeat the result in \cite{huang2024learning} for learning shallow unitaries. Be stab1 the family of single qubit stabilizer states $\mathrm{stab1}=\{\ket{0}, \ket{1}, \ket{+}, \ket{-}, \ket{y+}, \ket{y-}\}$, then:

\begin{lemma}[Lemma 10 in \cite{huang2024learning}, Learning a few-body observable with an unknown support]\label{lemma: huang shallow}
Given an error $\epsilon$, failure probability $\delta$, an unknown $n$-qubit observable $O$ with $\|O\|_{\infty} \leq 1$ that acts on an unknown set of $k$ qubits, and a dataset $\mathcal{T}_O(N) = \{|\psi_{\ell}\rangle = \bigotimes_{i=1}^{n} |\psi_{\ell,i}\rangle, v_{\ell}\}_{\ell=1}^{N}$, where $|\psi_{\ell,i}\rangle$ is sampled uniformly from \textup{stab1} and $v_{\ell}$ is a random variable with $\text{E}[v_{\ell}] = \langle\psi_{\ell}| O |\psi_{\ell}\rangle$, $|v_{\ell}| = \mathcal{O}(1)$. Given a dataset size of 
\begin{equation}
N = \frac{2^{\mathcal{O}(k) }\log(n/\delta)}{\epsilon^2},
\end{equation} 
with probability at least $1-\delta$, we can learn an observable $O'$ such that $\|O' - O\|_{\infty} \leq \epsilon$ and $\textup{supp}(O') \subseteq \textup{supp}(O)$. The computational complexity is $\mathcal{O}(n^k \log(n/\delta)/\epsilon^2)$.
\end{lemma}

The proof of Corollary~\ref{lemma: learning unitary} then goes by explicitly constructing a family of parametrized circuit $\{U(\bm{x})\}_{\bm{x}}$ and a set of input distributions $\{\mathcal{D}_i\}_i$ such that the concept class $\mathcal{M}_{U,W,M}$ is quantum learnable using Lemma~\ref{lemma: huang shallow} while still being classically hard. Let us rewrite here Corollary~\ref{lemma: learning unitary}

\shallowadv*
\begin{proof}
    Let first define the set of unitaries $\{U(\bm{x})\}_{\bm{x}}$, $\{W(\bm{\alpha})\}_{\bm{\alpha}}$, measurement operator $O$ and input distribution $\mathcal{D}$ on which the separation result holds. We provided a graphical representation of it in Figure \ref{fig:huang separation}.
    Be $n_S= \left \lfloor{\frac{n}{3}}\right \rfloor$, we define the set $\{U(\bm{x})\}_{\bm{x}}$ as the set of unitaries which act on a $n_S+1$ qubit system in the following way, depending on the first bit $x_1$ of the input $\bm{x}$. 
    
    If $x_1=0$ then each $U(\bm{x})$ is defined such that
    \begin{equation}
        U(\bm{x})\left(\ket{x_1}\otimes\ket{0^{n_S}}\right)= \ket{x_1}\otimes\ket{\psi_{stab}(\bm{x})}
    \end{equation}
    where $\ket{\psi_{stab}(\bm{x})}=\bigotimes_{i=1}^{n_S}\ket{\psi^i_{stab1}(\bm{x})}$ is a $n_S$ qubit tensor product of single qubit stabilizer states $\ket{\psi^i_{stab1}(\bm{x})}\in\{\ket{0},\ket{1},\ket{+},\ket{-},\ket{y+},\ket{y-} \}$. Since $n_S= \left \lfloor{\frac{n}{3}}\right \rfloor$, the output state $\ket{\psi_{stab}(\bm{x})}$ is completely described by the remaining $n-1$ input bits $x_2x_3...x_n$.
    
    Let consider now the case when $x_1=1$. Be $\mathcal{L}$ a $\BQP$-complete language defined over the input $\Tilde{\bm{x}}\in\{0,1\}^{n_S}$.
    Then we define the set of $\{U(\bm{x})\}_{\bm{x}}$ in the following way:
    \begin{equation}
        U(\bm{x})\left(\ket{x_1}\otimes\ket{0^{n_S}}\right)= \ket{x_1}\otimes\ket{\psi_{BQP}(\bm{x})}
    \end{equation}
    where $\ket{\psi_{BQP}(\bm{x})}$ is such that $\bra{\psi_{BQP}(\bm{x})}Z\otimes I \otimes ... \otimes I \ket{\psi_{BQP}(\bm{x})}=1/3$ if $\Tilde{\bm{x}}\in \mathcal{L}$ and $-1/3$ if $\Tilde{\bm{x}}\not \in \mathcal{L}$. As we described in the proof of Theorem~\ref{lemma: cl hardness}, such a circuit $U_{BQP}(\bm{x})$ always exists for any $\BQP$ language.

    As a set of unitaries $\{W(\bm{\alpha})\}_{\bm{\alpha}}$ we consider any set of parametrized unitaries acting on the $n_S$ qubit register such that for each $\bm{\alpha}$ parameter $W(\bm{\alpha})$ is shallow and such that $W(\bm{0})=I^{\otimes n_S}$. The $n_S$ qubit register is then measured by the observable $M=Z\otimes I \otimes ... \otimes I$.
    Finally we define the set of distributions $\{\mathcal{D}_i\}_i$ on the input bitstrings $\bm{x}\in\{0,1\}^n$ to be the following:
    \begin{itemize}
    \item The first bit $x_1\in \{0,1\}$ of $\bm{x}$ is extracted randomly with equal probability.
    \item If $x_1=0$ then the other $n-1$ bit $x_2x_3...x_n$ are extracted following the uniform distribution over $\{0,1\}^{n-1}$.
    \item If $x_1=1$ then the following $n_S$ bits $x_2x_3...x_{n_S+1}$ are extracted from an arbitrary distribution which characterizes $\calD_i$ for each index $i$. The other $n-(n_S+1)$ input bits are sampled uniformly at random. 
     
\end{itemize}
We now show that the concept class $\mathcal{M}_{U,W,M}$ defined in Eq. \ref{eq: concept unitary} with $\{U(\bm{x})\}_{\bm{x}}$, $\{W(\bm{\alpha})\}_{\bm{\alpha}}$ and $M$ considered above exhibits a learning separation with respect to the learning condition in Def.~\ref{def:learning}.
\vspace{4mm}

\textbf{Classical hardness} The argument is exactly the same as the one presented before in the proof of Theorem \ref{thm: separation W}.

\vspace{4mm}
\textbf{Quantum learnability} The quantum learnability is guaranteed by Lemma~\ref{lemma: huang shallow} in~\cite{huang2024learning}. Recall that the training data the learning algorithm receives for a concept $f^{\bm{\alpha}}\in \mathcal{M}_{U,W,M}$ are $\mathcal{T}^{\bm{\alpha}}=\{\bm{x}_\ell,y_\ell \}_\ell$ with $\mathbb{E}[y_\ell]=\text{Tr}[\rho_{U}(\bm{x}_\ell)O({\bm{\alpha}})]$. Now, in the case the first bit of $\bm{x}_\ell$ is 0, $\mathcal{T}^{\bm{\alpha}}$ is exactly the training set $\mathcal{T}_{\bm{\alpha}}$ required by the algorithm in Lemma \ref{lemma: huang shallow} to learn $O(\bm{\alpha})$. Since for every $\bm{\alpha}$ the unitary $W(\bm{\alpha})$ is of shallow depth, the locality of $O(\bm{\alpha})$ scales logarithmic with the number of qubit $n_S$.   
Then Lemma \ref{lemma: huang shallow} guarantees that the learning algorithm runs in polynomial time requiring a polynomial-sized dataset, a condition met as half of the training samples in $\mathcal{T}^{\bm{\alpha}}$ will suffice.

\end{proof}

\begin{figure}[h!]
    \includegraphics[width=\linewidth]{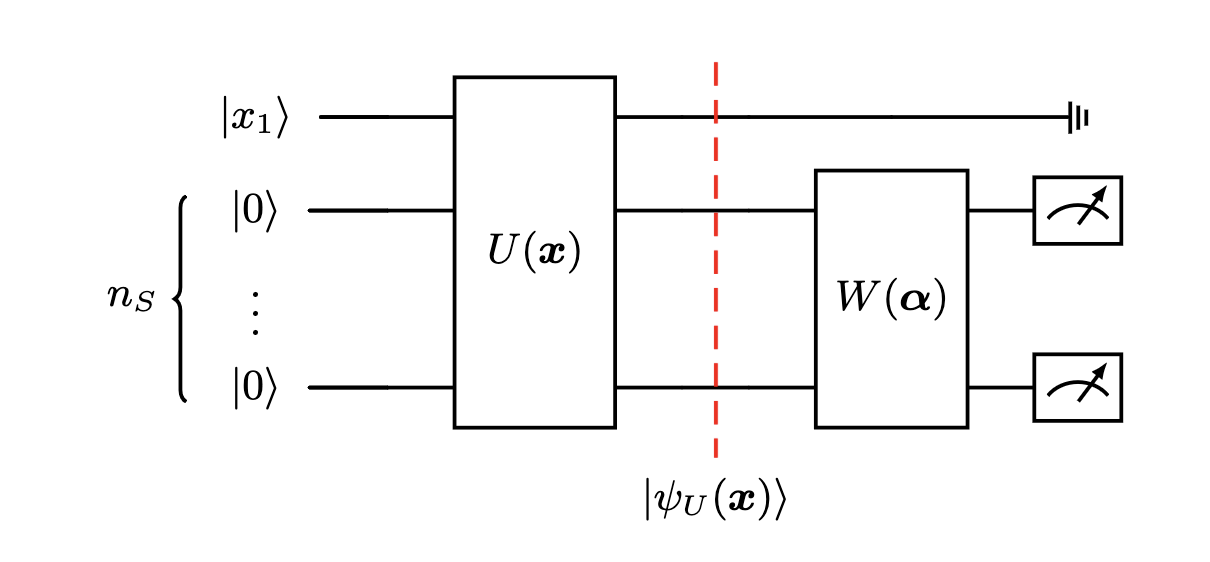}

    \caption{Quantum model which exhibits a learning separation for the concept class $\mathcal{M}_{U,W,O}$ when $W$ is a shallow circuit. The unitary $U(\bm{x})$ prepares the state $\ket{\psi_U(\bm{x})}=\ket{x_1}\otimes\ket{\psi^{n_S}_U(\bm{x_S})}$ where the form of $\ket{\psi^{n_S}_U(\bm{x_S})}$ depends on $x_1$, the first bit of each input $\bm{x}\in\{0,1\}^n$. If $x_1=0$, then $\ket{\psi^{n_S}_U(\bm{x_S})}=\ket{\psi_{stab}(\bm{x_S})}=\bigotimes_{i=1}^{n_S}\ket{\psi^i_{stab1}(\bm{x_S})}$ is the $n_S= \left \lfloor{\frac{n}{3}}\right \rfloor $ qubit tensor product of single-qubit stabilizers described by the classical bitstring $x_2x_3...x_n$. If $x_1=1$, then $\ket{\psi^{n_S}_U(\bm{x_S})}$ is the quantum state which decides the $n_S$ input bits $\bm{x_S}=x_2x_3...x_{n_S+1}$, considered as input of a $\BQP$-complete language $\mathcal{L}$ over the bitstrings $\bm{x_S}\in\{0,1\}^{n_S}$ . $W({\bm{\alpha}})$ is a parametrized shallow unitary, to prove classical hardness it is sufficient to consider $W({\bm{\alpha}}=0)=I^{\otimes n_S}$ and the measurement operator to be $O=Z\otimes I \otimes ... \otimes I$ on the $n_S$ qubits register.}
  \label{fig:huang separation}
\end{figure}

\bibliography{main.bib}

\end{document}